\documentclass{article}
\usepackage{geometry}
\usepackage[nospace,noadjust]{cite}
\usepackage{amsmath,amssymb,amsfonts,times}
\usepackage{graphicx}
\usepackage{textcomp}
\usepackage{xcolor}
\usepackage{subfigure}
\usepackage{epsfig}
\usepackage{multirow}
\usepackage{algorithm}
\usepackage{algorithmicx}
\usepackage{algpseudocode}
\usepackage{array}
\usepackage{epstopdf}
\usepackage{color}
\usepackage{diagbox}
\usepackage{bm}
\usepackage{bbm}
\usepackage{mathrsfs}
\usepackage{tcolorbox}
\usepackage{pdfsync}
\usepackage{mathtools}
\usepackage{makecell}

\usepackage{url}
\usepackage[hidelinks]{hyperref}
\usepackage{amsmath,amsthm,amssymb,amsbsy}
\usepackage{paralist}
\usepackage{xcolor}
\usepackage{color}
\usepackage{comment}
\usepackage{multirow}
\usepackage[toc, page]{appendix}

\usepackage{fancyhdr}
\usepackage{cite}
\usepackage{cleveref}
\newtheorem{lemma}{Lemma}%[section] %%    with section number.
\newtheorem{cor}{Corollary}%[section]
\newtheorem{theorem}{Theorem}

\theoremstyle{remark}

\newcommand{\R}{\mathbb{R}}

\newcommand{\C}{\mathbb{C}}

\newcommand{\e}{\begin{equation}}
\newcommand{\ee}{\end{equation}}
\newcommand{\en}{\begin{equation*}}
\newcommand{\een}{\end{equation*}}
\newcommand{\eqn}{\begin{eqnarray}}
\newcommand{\eeqn}{\end{eqnarray}}
\newcommand{\bmat}{\begin{bmatrix}}
\newcommand{\emat}{\end{bmatrix}}
\renewcommand{\Re}[1]{\operatorname{Re}\left\{#1\right\}}

\DeclareMathAlphabet\mathbfcal{OMS}{cmsy}{b}{n}
\renewcommand{\P}[1]{\operatorname{\mathbb{P}}\left(#1\right)}
\newcommand{\E}{\operatorname{\mathbb{E}}}

\newcommand{\vct}[1]{\boldsymbol{#1}}
\newcommand{\mtx}[1]{\boldsymbol{#1}}
\newcommand{\<}{\langle}
\renewcommand{\>}{\rangle}
\newcommand{\trace}{\operatorname{trace}}

\newcommand{\set}[1]{\mathbb{#1}}
\DeclareMathOperator*{\argmin}{\text{arg~min}}

\newcommand{\wh}{\widehat}

\newcommand{\wt}{\widetilde}
\newcommand{\ol}{\overline}

\newcommand{\calL}{\mathcal{L}}

\newcommand{\calN}{\mathcal{N}}

\newcommand{\calU}{\mathcal{U}}

\newcommand{\vv}{\vct{v}}

\newcommand{\vy}{\vct{y}}

\newcommand{\vtheta}{\vct{\theta}}

\newcommand{\vphi}{\vct{\phi}}
\newcommand{\vpsi}{\vct{\psi}}

\newcommand{\vxi}{\vct{\xi}}

\newcommand{\vrho}{\vct{\rho}}

\newcommand{\mA}{\mtx{A}}
\newcommand{\mB}{\mtx{B}}

\newcommand{\mH}{\mtx{H}}

\newcommand{\mP}{\mtx{P}}

\newcommand{\mU}{\mtx{U}}
\newcommand{\mV}{\mtx{V}}
\newcommand{\mW}{\mtx{W}}
\newcommand{\mX}{\mtx{X}}
\newcommand{\mY}{\mtx{Y}}
\newcommand{\mZ}{\mtx{Z}}

\newcommand{\mOmega}{\mtx{\Omega}}
\newcommand{\mPhi}{\mtx{\Phi}}

\newcommand{\mSigma}{\mtx{\Sigma}}

\newcommand{\mId}{{\bf I}}

\newcommand{\setR}{\set{R}}

\graphicspath{{./figs/}}

\newlength{\imgwidth}
\newboolean{twoColVersion}
\setboolean{twoColVersion}{false}
\newcommand{\twoCol}[2]{\ifthenelse{\boolean{twoColVersion}} {#1} {#2} }

\usepackage{mathtools}
\DeclarePairedDelimiter\ket{\lvert}{\rangle}

\title{\LARGE \bf Geometric Analysis of Variational Quantum Eigensolver}

\author{Zhen Qin\thanks{Zhen Qin is with the Michigan Institute for Computational Discovery and Engineering, Department of Electrical Engineering and Computer Science and Department of Statistics, University of Michigan, Ann Arbor, MI 48109 USA. (e-mail: zhenqin@umich.edu).}}

\begin{document}

\maketitle

\begin{abstract}
The Variational Quantum Eigensolver (VQE) is a fundamental algorithm in quantum computing, yet its theoretical understanding remains fragmented across fixed-ansatz and adaptive-circuit formulations, which are typically analyzed using different frameworks. As a result, a coherent geometric characterization of the VQE remains missing. In this paper, we establish a geometric analysis of VQE in terms of optimization landscape, initialization guarantee, and noise robustness. First, from a Riemannian geometric perspective, we study the optimization landscape via an ansatz-free product-unitary formulation over the unitary group, where each layer unitary is treated as an independent optimization variable, thereby capturing both fixed-ansatz and adaptive-circuit paradigms. For the single-unitary case, we establish explicit linear convergence guarantees for Riemannian gradient descent (RGD) under proper initialization and characterize the global landscape by proving that every critical point is either a global minimum or a strict saddle point. For the product-unitary case, we extend the convergence analysis of RGD and show that the convergence rate deteriorates polynomially with circuit depth, providing a convergence-theoretic explanation of the barren plateau (BP) phenomenon through the interplay between Hilbert-space geometry, target-unitary complexity, per-layer expressivity, and circuit depth. Second, we analyze initialization guarantees for the above convergence results. We consider a practical variational circuit composed of small-angle random Pauli rotations and prove that this initialization scheme satisfies the required geometric conditions with high probability. We further characterize how the resulting initialization error depends on the reference state, circuit depth, and random angles. Third, we study robustness under finite-shot measurements for VQE with Pauli-decomposed Hamiltonians. We show that RGD retains linear convergence up to a noise-dominated neighborhood of the global minimizer. Moreover, we extend the analysis to non-uniform measurement allocation under a fixed measurement budget and show that coefficient-adaptive allocation achieves strictly lower statistical error than uniform sampling. Together, these results provide a systematic geometric characterization of the optimization, initialization, and measurement behavior of VQE, offering new insights into its trainability and robustness.
\end{abstract}

%\begin{IEEEkeywords}
%VQE, BP, landscape analysis, initialization, noise robustness.
%\end{IEEEkeywords}

\section{Introduction}

The Variational Quantum Eigensolver (VQE) \cite{peruzzo2014variational,mcclean2016theory,kandala2017hardware,wang2019accelerated,tilly2022variational} has emerged as one of the most prominent quantum algorithms in the Noisy Intermediate-Scale Quantum (NISQ) era. By leveraging the variational principle \cite{rayleigh1870finding,weber2004essential}, VQE provides a hybrid quantum-classical framework for approximating the ground-state energy of a target Hamiltonian, thereby yielding an upper bound to its exact ground-state energy.
Since ground-state properties form the foundation for understanding the energetic and dynamical behavior of quantum many-body systems, efficient ground-state preparation constitutes a central challenge in both quantum chemistry \cite{deglmann2015application,williams2018free,heifetz2020quantum} and condensed matter physics \cite{van2020rechargeable,xu2020test,continentino2021key}. Beyond its theoretical significance, electronic-structure computation underpins a wide range of practical applications, including drug discovery \cite{cao2018potential,blunt2022perspective,mustafa2022variational}, materials science \cite{lordi2021advances,mihalikova2022cost,lee2025variational}, and chemical engineering \cite{cao2019quantum,ajagekar2022new,de2022optimization}. The increasing complexity of these applications demands highly accurate simulations of large-scale quantum systems, posing a major challenge for conventional computational methods. In particular, classical approaches often become computationally prohibitive due to the exponential growth of the underlying Hilbert space with system size. As a hybrid quantum-classical framework, VQE offers a promising alternative by employing parameterized quantum circuits to efficiently represent many-body wavefunctions, thereby opening a potential pathway toward scalable quantum simulations beyond classical tractability.

The central objective of the VQE is to approximate the ground-state energy of a target Hamiltonian $\mH$ by minimizing the expectation value of $\mH$ over a parameterized family of quantum states. Specifically, the trial state is prepared by applying a parameterized quantum circuit to a reference state $|\vphi_0\rangle$, which is successively transformed by a sequence of $N$ unitary operators $\mU_1(\theta_1), \ldots, \mU_N(\theta_N)$, where $N$ denotes the circuit depth. The resulting variational objective is formulated as
\begin{eqnarray}
\label{objective function in the introduction parameterized}
G(\theta_1,\ldots,\theta_N) = \langle \vphi_0 | \mU_N^\dagger(\theta_N)\cdots \mU_1^\dagger(\theta_1) \mH \mU_1(\theta_1)\cdots \mU_N(\theta_N) |\vphi_0\rangle.
\end{eqnarray}
While this parameterized formulation is well suited for near-term quantum hardware, optimizing the variational objective remains highly nontrivial. In particular, the optimization landscape often suffers from the barren plateau (BP) phenomenon \cite{wecker2015progress,mcclean2018barren,bittel2021training,cerezo2021cost,larocca2025barren}, where the variance of the gradients decays exponentially with the number of qubits. Consequently, the landscape becomes nearly flat over most of the parameter space, so that small perturbations of the circuit parameters induce only exponentially small changes in the objective value. As a result, optimization algorithms may fail to extract reliable descent directions from local objective evaluations, significantly degrading trainability and making the optimization process increasingly difficult and highly sensitive to the choice of ansatz. This phenomenon has been attributed to several interrelated factors, including the exponentially large dimension of the Hilbert space, which induces an intrinsic curse of dimensionality in quantum state representations \cite{volkoff2021large,diaz2023showcasing,ragone2024lie,fontana2024characterizing,cerezo2025does}; the circuit architecture and depth, where increased expressivity of deeper ansatzes often comes at the cost of more severe optimization landscapes such as barren plateaus \cite{mcclean2018barren,cerezo2021cost,larocca2022diagnosing,holmes2022connecting}; the choice of parameter initialization, which can strongly affect the optimization trajectory and lead to vanishing gradients or slow convergence \cite{grant2019initialization,patti2021entanglement,sauvage2021flip,rad2022surviving,zhang2022escaping,wang2024trainability,puig2025variational}; and the presence of hardware noise, which introduces stochastic perturbations that further degrade trainability and distort gradient information \cite{wang2021noise,stilck2021limitations,de2023limitations,schumann2024emergence,sannia2024engineered,liu2025stochastic,singkanipa2025beyond}.

In light of these challenges, there has been growing interest in adopting a geometric perspective on quantum optimization. Rather than restricting the search to a fixed parameterized ansatz, one may instead optimize directly over the unitary group, thereby removing the architectural constraints imposed by a prescribed circuit structure. This lifts the optimization problem from a finite-dimensional parameter space to the underlying manifold of unitary operators and enables the use of Riemannian optimization techniques. Specifically, the variational problem can be formulated as the following manifold optimization problem:
\begin{eqnarray}
\label{loss function of one layer introduction}
f(\mU) =  \langle\vphi_0|\mU^\dagger\mH \mU |\vphi_0 \rangle,
\end{eqnarray}
where $\mU$ lies in the unitary group $\calU(D) := \{\mU\in\C^{D\times D}:  \mU\mU^\dagger = \mU^\dagger\mU = \mId   \}$.  This formulation admits a broad class of Riemannian optimization methods, including both first-order and second-order Riemannian gradient descent (RGD) algorithms \cite{wiersema2023optimizing,magann2023randomized,malvetti2024randomized,mcmahon2025equating,pervez2025riemannian,lai2026quantum}. Although the optimization is performed directly on the unitary manifold rather than within a fixed parametrization, it remains physically implementable through an adaptive circuit construction: at each iteration, the information extracted from the current quantum state is used to identify additional gates that further reduce the objective value. By iteratively appending such gates, the circuit is progressively refined toward the target ground state. Unlike conventional VQE, where the ansatz is specified a priori, the circuit architecture here evolves dynamically during optimization. Such dynamically constructed ansatzes are commonly referred to as adaptive circuits, as exemplified by AdaptVQE \cite{grimsley2019adaptive,tang2021qubit}.

For both the fixed-ansatz objective \eqref{objective function in the introduction parameterized} and the adaptive-circuit objective \eqref{loss function of one layer introduction}, extensive geometric studies have been conducted to characterize their optimization landscapes. These analyses can be broadly organized according to the underlying optimization formulation:
\begin{itemize}

\item \textbf{Fixed-ansatz objective.}
For the parameterized objective \eqref{objective function in the introduction parameterized}, existing geometric landscape analyses mainly fall into two categories:
\begin{itemize}
    \item \textit{Probabilistic concentration.}
    Refs.~\cite{mcclean2018barren,cerezo2021cost} show that both the loss function and its gradients become exponentially concentrated around their expectations as the number of qubits $n$ increases. Consequently, the optimization landscape is exponentially flat over most of the parameter space, so that small perturbations of the model parameters $\vtheta=(\theta_1,\dots,\theta_N)$ induce only exponentially small variations in the loss and gradients. This characterization is inherently statistical: while it implies that most regions of the landscape exhibit exponentially small variations, it does not preclude the existence of localized regions with non-negligible gradients, whose geometric structure remains largely uncharacterized.

    \item \textit{Deterministic concentration.}
    Refs.~\cite{wang2021noise,diaz2023showcasing,thanasilp2023subtleties,leone2024practical} establish a stronger characterization by showing that the deviation of the loss function from its expectation is uniformly bounded by an exponentially small quantity for all parameter values $\vtheta$. As a result, the entire landscape is uniformly flat, with exponentially vanishing gradients everywhere. Such guarantees, however, typically hold only under specific settings, such as noisy circuits or particular objective structures.
\end{itemize}

\item \textbf{Adaptive-circuit objective.}
For the adaptive-circuit objective \eqref{loss function of one layer introduction}, existing geometric landscape analyses mainly fall into two categories:
\begin{itemize}
    \item \textit{Convergence analysis.}
    Refs.~\cite{magann2023randomized,pervez2025riemannian,mcmahon2025equating} analyze the convergence behavior of RGD-based methods for ground-state preparation and quantum imaginary time evolution, but do not establish explicit quantitative convergence rates to the target ground state.

    \item \textit{Global landscape analysis.}
    Ref.~\cite{malvetti2024randomized} proves almost-sure convergence of a randomized Riemannian gradient method from almost all initializations to a global minimizer by leveraging the strict-saddle landscape characterization developed in \cite{panageas2019first}. However, this guarantee is algorithm-dependent and relies essentially on randomization to escape saddle points, and therefore does not directly extend to general deterministic RGD methods.
\end{itemize}
\end{itemize}

The above discussion shows that existing geometric analyses of VQE are typically developed under different optimization formulations and therefore capture only partial aspects of the optimization landscape. As a result, a coherent and systematic understanding of VQE from a geometric perspective remains largely lacking. In this paper, we establish a rigorous geometric framework for analyzing VQE and summarize our main contributions as follows:
\begin{itemize}
\item \textbf{Landscape analysis:} We initiate our study by analyzing the VQE landscape. To characterize the optimization landscape under both adaptive-circuit and fixed-ansatz formulations, we adopt a Riemannian optimization perspective and consider the following ansatz-free formulation over the unitary group:
    \begin{eqnarray}
    \label{objective function in the introduction}
    g(\mU_1,\ldots,\mU_N) =  \langle \vphi_0 | \mU_N^\dagger \cdots \mU_1^\dagger \mH \mU_1 \cdots \mU_N | \vphi_0 \rangle,
    \end{eqnarray}
    where each $\mU_k$ is treated as an independent optimization variable in $\calU(D)$, rather than being constrained by a prescribed circuit parametrization. This formulation provides a unified framework for the two objectives discussed above, by removing parameterization-induced couplings in \eqref{objective function in the introduction parameterized} and recovering the adaptive-circuit formulation \eqref{loss function of one layer introduction} as the special case $N=1$. This perspective underlies the following two results.
    \begin{itemize}

    \item For the single-unitary objective ($N=1$), we establish an explicit local convergence guarantee for RGD, including a concrete linear convergence rate and a quantitative characterization of its basin of attraction. We further show that the objective satisfies the strict saddle property: every critical point is either a global minimizer or a strict saddle point.

    \item For the product-unitary objective ($N>1$), we prove that RGD converges linearly to a global minimizer under proper initialization. In addition, we show that the convergence rate deteriorates polynomially with the circuit depth, thereby providing a convergence-theoretic characterization of the BP phenomenon through the interplay between Hilbert-space geometry, target-unitary complexity, per-layer expressivity, and circuit depth.
    \end{itemize}

\item \textbf{Initialization guarantee:} To ensure the initialization conditions required for local convergence, we analyze a practical variational quantum circuit composed of parameterized Pauli rotations with small random angles. We show that this initialization scheme satisfies the required conditions with high probability, and further characterize how the resulting initialization error depends on the reference state, circuit depth, and random initialization angles.

\item \textbf{Noise robustness:} We analyze the robustness of VQE under finite-shot Pauli measurements arising from the Hamiltonian decomposition. We show that RGD retains linear convergence up to a noise-dominated neighborhood of the global minimizer, where the residual error is determined by the number of measurement shots allocated to each Pauli operator. Furthermore, under a fixed total measurement budget, we derive an optimal coefficient-adaptive shot allocation strategy and prove that it achieves strictly smaller statistical error than uniform measurement allocation.

\end{itemize}

{\bf Notation}: We use bold capital letters (e.g., $\mY$) to denote matrices,  bold lowercase letters (e.g., $\vy$) to denote vectors, and italic letters (e.g., $Y$ or $y$) to denote scalar quantities. Bra--ket notation is used to denote quantum states: $|\vphi\rangle$ represents a (column) state vector, and $\langle\vphi|$ its Hermitian transpose.  Matrix elements are denoted using subscripts. For example, $Y_{s_1s_2}$ denotes the element in position $(s_1, s_2)$ of the matrix $\mY$.
$\|\mX\|$ and $\|\mX\|_F$ respectively represent the spectral norm and Frobenius norm of $\mX$. The superscript $(\cdot)^\dagger$ denotes the Hermitian transpose. For a positive integer $K$, $[K]$ denotes the set $\{1,\dots, K \}$.

\section{Landscape Analysis of the Variational Quantum Eigensolver on the Unitary Group}
\label{Sec: landscape analysis of VQE}

In this section, we analyze the optimization landscape of the VQE on the unitary group, viewed as a Riemannian manifold. We investigate both local and global geometric properties of the objective function in \eqref{objective function in the introduction}. Our analysis is carried out in two stages, moving from the single-unitary formulation to the more general product-unitary formulation. First, we consider the single-unitary setting, which serves as the foundational case for the subsequent product-unitary analysis. We show that, under a sufficiently accurate initialization, Riemannian gradient descent (RGD) converges linearly to a global minimizer, with a rate determined by the spectral gap and the operator norm of the Hamiltonian. Moreover, in the absence of such initialization guarantees, we characterize the global landscape and show that all critical points are either strict saddle points or global minima. Second, we extend the analysis to the product-unitary setting. We establish that similar local convergence guarantees hold under appropriate initialization, and we further discuss the corresponding global landscape structure.

\subsection{Single-Unitary Formulation}
\label{subsec One-Layer Quantum Circuit}

As a preliminary step toward analyzing the product-unitary formulation, we first consider the reduced single-unitary setting. Specifically, by setting $N=1$ in \eqref{objective function in the introduction}, we identify
\begin{eqnarray}
\label{product of unitary matrices}
\mU := \mU_1,
\end{eqnarray}
which serves as the fundamental building block of the general product-unitary model. For convenience, we recall the corresponding single-unitary objective:
\begin{eqnarray}
\label{loss function of one layer}
f(\mU) = \trace(\mH \mU \vrho_0 \mU^\dagger) = \langle\vphi_0|\mU^\dagger\mH \mU |\vphi_0 \rangle,
\end{eqnarray}
where the unitary matrix $\mU \in \calU(D) := \{\mU\in\C^{D\times D}:  \mU\mU^\dagger = \mU^\dagger\mU = \mId   \}$, which denotes the unitary group and coincides with the Stiefel manifold $\text{St}(D,D)$ in the square case. The initial state is given by $\vrho_0 = |\vphi_0\rangle \langle \vphi_0| \in \mathbb{C}^{D \times D}$, and without loss of generality, the Hamiltonian is defined as $\mH = \sum_{k=0}^{D-1}E_k |\vpsi_k \rangle\langle \vpsi_k|$,
where $E_0 \leq E_1 \leq \cdots \leq E_{D-1}$ are the eigenvalues of $\mH$  and $\{|\vpsi_k\rangle\}_{k=0}^{D-1}$ are the corresponding
orthonormal eigenvectors. We define the spectral gap as $\Delta_1 \coloneqq E_s - E_0 > 0$, where $s = \min\{k : E_k > E_0\}$ is the index of the first eigenvalue strictly greater than $E_0$. Consequently, the ground-state eigenspace may be degenerate, i.e., $E_0 = \cdots = E_{s-1}< E_{s}\leq \cdots \leq E_{D-1}$.

We now define the optimal solution of the optimization problem~\eqref{loss function of one layer} over the unitary group:
\begin{eqnarray}
\label{loss function of one layer optimal}
\mU^\star := \argmin_{\mU\in\calU(D)}f(\mU).
\end{eqnarray}
By the variational principle \cite{rayleigh1870finding,weber2004essential}, we have $f(\mU) \geq E_0$ for all $\mU \in \calU(D)$. Moreover, since $\calU(D)$ acts transitively on the unit sphere, any pure state can be generated from $|\vphi_0\rangle$, implying that there exists $\mU^\star \in \calU(D)$ such that $\mU^\star |\vphi_0\rangle$ lies in the ground-state eigenspace of $\mH$. Consequently, the global minimum is attainable and satisfies $f(\mU^\star) = E_0$. We next study the local geometric properties of the optimization landscape in a neighborhood of the global minimizer.

\paragraph{Local Convergence Analysis.} We employ RGD~\cite{wiersema2023optimizing,magann2023randomized,malvetti2024randomized,mcmahon2025equating,pervez2025riemannian,lai2026quantum} on the unitary group to optimize~\eqref{loss function of one layer}, which naturally preserves the unitary structure throughout the optimization process. We then analyze the corresponding RGD dynamics to characterize its local convergence behavior. Specifically, the update is given by
\begin{eqnarray}
\label{iteration of RGD on the Stiefel}
\mU^{t+1} = \text{Retr}_{\mU}( - \mu \text{grad}_{\mU}f(\mU^t) ),
\end{eqnarray}
where $\mu > 0$ is the step size. The mapping $\text{Retr}_{\mU}(\vxi)\coloneqq (\mU + \vxi)\big((\mU + \vxi)^\dagger(\mU + \vxi)\big)^{-\frac{1}{2}}$ denotes
the retraction\footnote{For theoretical convenience, we adopt the polar-based retraction in the analysis. In practical quantum circuit implementations, this retraction can be equivalently realized via the exponential map \cite{lee2003smooth} or its first-order Trotter approximations \cite{trotter1959product,lloyd1996universal}; see \cite{wiersema2023optimizing,lai2026quantum} for implementation details.} onto the unitary group. To derive the Riemannian gradient $\text{grad}_{\mU}f(\mU^t)$, we first characterize the tangent space of the unitary group. The tangent space at $\mU\in\C^{D\times D}$ is given by
\begin{eqnarray}
\text{T}_{\mU}\calU(D) \coloneqq \{\mOmega\mU \in \mathbb{C}^{D\times D}:
\mOmega^\dagger = -\mOmega\},
\end{eqnarray}
i.e., the set of matrices of the form $\mOmega\mU$, where $\mOmega$ is skew-Hermitian. With this characterization, the Riemannian gradient\footnote{By introducing the matrix commutator $[\mA,\mB] := \mA\mB - \mB\mA$, \eqref{gradient of RGD on the Stiefel} can be rewritten as $\text{grad}_{\mU}f(\mU^t) = \frac{1}{2}[\mH, \mU^t \vrho_0 (\mU^t)^\dagger ]\mU^t$.} is obtained by orthogonally projecting the Euclidean gradient $\nabla_{\mU}f(\mU^t) = \mH\mU^t\vrho_0$ onto $\text{T}_{\mU}\calU(D)$, yielding
\begin{eqnarray}
\label{gradient of RGD on the Stiefel}
\text{grad}_{\mU}f(\mU^t) &\!\!\!\!=\!\!\!\!& \mH \mU^t \vrho_0 - \mU^t\text{sym}({\mU^t}^\dagger\mH\mU^t\vrho_0),
\end{eqnarray}
where $\text{sym}(\mA) \coloneqq \frac{\mA + \mA^\dagger}{2}$ denotes the Hermitian part of $\mA$.

Building upon this Riemannian gradient formulation and its associated geometric structure, we establish the following local convergence guarantee.
\begin{theorem}
\label{Theorem: convergence analysis of single circuit}
Suppose that the initialization $\mU^0$ satisfies $f(\mU^0) - f(\mU^\star)\leq \frac{\Delta_1}{2}$. where $\Delta_1 = E_s - E_0>0$ denotes the spectral gap. Let $\{\mU^t\}_{t\geq 1}$ be the sequence generated by Riemannian gradient descent applied to \eqref{loss function of one layer} with step size $\mu \leq \frac{1}{9\|\mH\|}$. Then, the iterates converge linearly to $\mU^\star$ in the sense that
\begin{eqnarray}
\label{final conclusion of single circuit main theorem}
f(\mU^{t} ) - f(\mU^\star) \leq \bigg( 1 - \frac{\Delta_1^2\mu}{32\|\mH\|} \bigg)^t\big(f(\mU^0 ) - f(\mU^\star)\big).
\end{eqnarray}
\end{theorem}
The proof is provided in {Appendix}~\ref{Proof of convergence analysis of single circuit appendix}. From \Cref{Theorem: convergence analysis of single circuit}, by choosing $\mu = \frac{1}{9\|\mH\|}$, the convergence rate is specified as $1 - \frac{\Delta_1^2}{288\|\mH\|^2}$. This result establishes the existence of a basin of attraction around the optimum, within which the iterates converge linearly. In particular, the convergence rate depends explicitly on the spectral gap $\Delta_1$ and the operator norm $\|\mH\|$, and deteriorates as $\Delta_1$ decreases relative to $\|\mH\|$. This behavior reflects the intrinsic difficulty of optimizing quantum systems with nearly degenerate ground states. Overall, this establishes the existence of a well-behaved local landscape around the global optimum, within which first-order Riemannian methods exhibit provable linear convergence. Compared with the previous convergence analyses in \cite{magann2023randomized,pervez2025riemannian,mcmahon2025equating}, our result provides an explicit local convergence characterization by establishing a concrete linear convergence rate together with an explicit basin of attraction around a global minimizer.

\paragraph{Global Landscape Structure.} Beyond the benign local region discussed above, existing work \cite{malvetti2024randomized} shows that a randomized Riemannian gradient method converges to a global minimizer from almost all initializations; however, this result is inherently algorithm-dependent and does not directly characterize the intrinsic geometry of the objective. A fundamental open question is whether the global landscape of $f(\mU)$ admits a favorable geometric structure independent of any particular optimization algorithm. In general, non-convex optimization problems \cite{zhu2018global,haeffele2019structured,chi2019nonconvex,zhu2021global} may possess local minima that are suboptimal, making it impossible to guarantee convergence to a global solution from random initialization via gradient descent. Despite this general difficulty, we show that the objective function $f(\mU)$ in \eqref{loss function of one layer} exhibits a remarkably benign global geometric structure: every critical point is either a global minimum or a strict saddle point, i.e., a point at which the Riemannian Hessian possesses at least one strictly negative eigenvalue.  We formalize this result in the following theorem.
\begin{theorem}
\label{Theorem: global optimization of one-layer quantum circuit}
Let $f(\mU) = \langle\vphi_0|\mU^\dagger\mH \mU |\vphi_0 \rangle$ defined in \eqref{loss function of one layer}, where $\mU \in \calU(D)$ and the Hamiltonian $\mH$ admits the spectral decomposition $\mH = \sum_{k=0}^{D-1}E_k|\vpsi_k\rangle\langle\vpsi_k|$. We assume that the eigenvalues satisfy $E_0 = \cdots = E_{s-1} < E_s \leq \cdots \leq E_{D-1}$, and define the spectral gap $\Delta_1 = E_s - E_0 > 0$. Then $f(\mU)$ has exactly $D$ critical points $\{\wh{\mU}_k\}_{k=0}^{D-1}$, characterized by $\wh{\mU}_k|\vphi_0\rangle = |\vpsi_k\rangle$, whose Riemannian Hessian bilinear form satisfies
\begin{eqnarray}
    \text{Hess}\,f(\wh{\mU}_k)[\mOmega\wh{\mU}_k, \mOmega\wh{\mU}_k]
    = 2\sum_{l=0}^{D-1}(E_l - E_k)|\Omega_{lk}|^2,
\end{eqnarray}
where $\mOmega $ is skew-Hermitian with matrix elements $\Omega_{jl} := \langle\vpsi_j|\mOmega|\vpsi_l\rangle$, and $\mOmega\wh{\mU}_k \in T_{\mU}\calU(D)$. In particular:
\begin{enumerate}
    \item \textbf{(Global minima.)} The $s$ critical points $\{\wh{\mU}_k\}_{k=0}^{s-1}$, each achieving $f(\wh{\mU}_k) = E_0$, are global minima of $f$. Their Hessians are positive semidefinite:
    \begin{eqnarray}
        \text{Hess}\,f(\wh{\mU}_k)[\mOmega\wh{\mU}_k, \mOmega\wh{\mU}_k]
        = 2\sum_{l=s}^{D-1}(E_l-E_0)|\Omega_{lk}|^2
        \geq 2\Delta_1\sum_{l=s}^{D-1}|\Omega_{lk}|^2 \geq 0.
    \end{eqnarray}

    \item \textbf{(Strict saddle points.)} Every critical point $\wh{\mU}_k$ with $k \geq s$ is a strict saddle point. Specifically, for any $l < s$, choosing $\mOmega$ with $\Omega_{lk} \neq 0$ yields
    \begin{eqnarray}
        \text{Hess}\,f(\wh{\mU}_k)[\mOmega\wh{\mU}_k, \mOmega\wh{\mU}_k]
        \leq 2(E_l - E_k)|\Omega_{lk}|^2 < 0,
    \end{eqnarray}
    confirming that the Hessian has at least one strictly negative eigenvalue, i.e., a strict descent direction exists.
\end{enumerate}
Consequently, $f$ satisfies the \emph{strict saddle property} on $\mathcal{U}(D)$: every critical point is either a global minimum or a strict saddle point.
\end{theorem}
The proof is provided in Appendix~\ref{Proof of global optimization}. Theorem~\ref{Theorem: global optimization of one-layer quantum circuit} establishes that the global optimization landscape of $f(\mU)$ is geometrically benign in the sense that gradient-based methods cannot be permanently trapped at suboptimal critical points. However, this result does not guarantee an efficient convergence rate of RGD from random initialization. As shown in the following discussion, the Riemannian gradient can become exponentially small when the iterate is far from the ground-state subspace, significantly slowing down progress of gradient-based methods. Therefore, despite the absence of spurious local minima, efficient optimization still requires a suitably chosen initialization to avoid regions with vanishing gradients.

\paragraph{Mechanism of Vanishing Gradients.}
Despite the globally benign landscape, this does not guarantee efficient convergence from random initialization. In particular, optimization may still suffer from barren plateaus, where the gradient becomes vanishingly small even when the state remains far from optimal. To characterize the conditions under which such unfavorable regimes arise, we derive a lower bound on the Frobenius norm of the Riemannian gradient. Let $\ket{\vphi} = \mU\ket{\vphi_0} = \sum_{k=0}^{D-1} c_k \ket{\vpsi_k}$, where the coefficients satisfy $\sum_{k=0}^{D-1} |c_k|^2 = 1$. We define the excited-space weight on the orthogonal complement of the ground-state subspace as
\begin{equation}
p \;:=\; \sum_{k=s}^{D-1}|c_k|^2,
\qquad
1-p \;:=\; \sum_{k=0}^{s-1}|c_k|^2.
\end{equation}
Following the derivation in \eqref{lower bound of F in the single layer} of {Appendix}~\ref{Proof of convergence analysis of single circuit appendix},  we have
\begin{equation}\label{eq:grad-lb-p}
\|\text{grad}_{\mU} f(\mU)\|_F^2
\;\ge\; \frac{\Delta_1^2}{4}\,p(1-p)
\;=\; \frac{\Delta_1^2}{4}
\bigg(\sum_{k=s}^{D-1}|c_k|^2\bigg)
\bigg(\sum_{k=0}^{s-1}|c_k|^2\bigg).
\end{equation}
The lower bound in \eqref{eq:grad-lb-p} shows that the gradient norm is controlled by the product $p(1-p)$. In particular, it can be small when $p(1-p)$ is small, which occurs in two qualitatively distinct regimes:
\begin{itemize}
    \item \textbf{Case I (near the ground-state subspace): $p \ll 1$.}
In this regime, $\ket{\vphi}$ has almost all of its weight in the ground-energy subspace $\mathrm{span}\{\ket{\vpsi_0},\dots,\ket{\vpsi_{s-1}}\}$,
i.e., $\|(I-P_0)\ket{\vphi}\|^2 = \|\sum_{k=s}^{D-1} c_k \ket{\vpsi_k}\|^2 = \sum_{k=s}^{D-1}|c_k|^2 = p$ is small, where $P_0:=\sum_{j=0}^{s-1}|\vpsi_j\>\<\vpsi_j|$.
Consequently, the suboptimality in energy is also small (indeed, $\Delta_1 p \le f(\mU)-f(\mU^\star)$ as shown in \eqref{lower bound of p 1} of {Appendix}~\ref{Proof of convergence analysis of single circuit appendix}), so a small gradient is consistent with being
already close to the optimum and does not indicate a barren-plateau-type obstruction.
\item \textbf{Case II (almost orthogonal to the ground-state subspace): $1-p \ll 1$.}
In this regime, the overlap with the ground-energy subspace is exponentially small,
$\|P_0\ket{\vphi}\|^2 = \|\sum_{k=0}^{s-1} c_k \ket{\vpsi_k}\|^2 = \sum_{k=0}^{s-1}|c_k|^2 = 1-p \approx 0$, and the state is almost entirely supported on excited energies.
In contrast to Case~I, the suboptimality $f(\mU)-f(\mU^\star)$ may still be $O(1)$ (or larger), while
\eqref{eq:grad-lb-p} only guarantees
\begin{equation}
\|\text{grad}_{\mU} f(\mU)\|_F^2 \;\gtrsim\; \frac{\Delta_1^2}{4}(1-p),
\end{equation}
which itself becomes very small when $1-p \ll 1$.
This is precisely the regime in which one may observe extremely small gradients despite being far from optimality, consistent with the barren-plateau phenomenon.
\end{itemize}

We now relate the above discussion to the initialization condition in \Cref{Theorem: convergence analysis of single circuit}. In particular, using the relation $p \le (f(\mU^{(0)})-f(\mU^\star))/\Delta_1$ from \eqref{lower bound of p 1} in Appendix~\ref{Proof of convergence analysis of single circuit appendix}, the initialization condition $f(\mU^{(0)})-f(\mU^\star)\le \Delta_1/2$ implies that $p \le 1/2$, and hence $1-p \ge 1/2$. This guarantees that the initial point lies in a favorable region where the overlap with the ground-energy subspace is sufficiently large, thereby excluding the barren-plateau regime associated with Case II. At the same time, this condition places the iterate within the basin of attraction characterized in \Cref{Theorem: convergence analysis of single circuit}, where RGD enjoys linear convergence. Therefore, a suitable initialization simultaneously avoids the vanishing-gradient regime and ensures that the optimization proceeds in a region with provable linear convergence behavior.

\subsection{Product-Unitary Formulation}

We now turn to the product-unitary formulation and study the BP phenomenon from a convergence-theoretic perspective. Unlike existing BP analyses \cite{mcclean2018barren,cerezo2021cost,wang2021noise,diaz2023showcasing,thanasilp2023subtleties,leone2024practical}, which are primarily based on concentration arguments, our analysis is based on deterministic convergence theory and characterizes the trainability of the product-unitary objective through its optimization dynamics. To this end, we consider the product-unitary parameterization $\mU_1 \mU_{2} \cdots \mU_N$, where each $\mU_h \in \calU(D)$ lies on the unitary group $\calU(D)$. This layered structure introduces additional geometric interactions between the unitary factors, which can significantly affect the optimization landscape and convergence behavior compared to the single-unitary case. For convenience, we recall the objective function in \eqref{objective function in the introduction}:
\begin{eqnarray}
\label{loss function of N layer}
g(\mU_1,\ldots,\mU_N)
&\!\!\!\!=\!\!\!\!& \trace\big(\mH \mU_1 \cdots \mU_N \vrho_0 \mU_N^\dagger \cdots \mU_1^\dagger\big) \nonumber\\
&\!\!\!\!=\!\!\!\!& \langle \vphi_0 | \mU_N^\dagger \cdots \mU_1^\dagger \mH \mU_1 \cdots \mU_N | \vphi_0 \rangle.
\end{eqnarray}
By the variational principle, we have $f(\mU_1,\ldots,\mU_N) \geq E_0$. Moreover,  since the product-unitary representation $\mU_1 \cdots \mU_N$ is expressive enough to reach any unitary in $\calU(D)$, the ground state remains attainable within this parameterization. Therefore, the optimal solution
\begin{eqnarray}
\label{optimal solution of N layers loss}
(\mU_1^\star,\ldots,\mU_N^\star) := \argmin_{\mU_1,\ldots,\mU_N} g(\mU_1,\ldots,\mU_N)
\end{eqnarray}
achieves the optimal value
\begin{eqnarray}
\label{optimal solution of N layers in loss}
g(\mU_1^\star,\ldots,\mU_N^\star) = E_0.
\end{eqnarray}

\paragraph{Local Convergence Analysis.} Similar to the single-unitary setting, we can apply RGD to analyze the geometric properties of the product-unitary objective \eqref{loss function of N layer}. Specifically, the update for each unitary factor $h$ is given by
\begin{eqnarray}
\label{iteration of RGD on the Stiefel N layer}
\mU_h^{t+1} = \text{Retr}_{\mU_h}( - \mu \text{grad}_{\mU_h}g(\mU_1^t,\dots, \mU_N^t) )
\end{eqnarray}
where $\mu>0$ is the step size. The Riemannian gradient  for the $h$-th unitary factor is
\begin{eqnarray}
\label{Riemannian gradient on the Stiefel N layer}
\text{grad}_{\mU_h}g(\mU_1^t,\dots, \mU_N^t) =  \nabla_{\mU_h}g(\mU_1^t,\ldots,\mU_N^t) - \mU_h^t\text{sym}({\mU_h^t}^\dagger\nabla_{\mU_h}g(\mU_1^t,\ldots,\mU_N^t)) ,
\end{eqnarray}
where the Euclidean gradient is $\nabla_{\mU_h}g(\mU_1^t,\ldots,\mU_N^t) = (\mU_1^t\cdots \mU_{h-1}^t)^\dagger\mH\mU_1^t\cdots \mU_N^t\vrho_0(\mU_{h+1}^t\cdots \mU_{N}^t)^\dagger$.

Leveraging the Riemannian geometric structure of the product-unitary objective, we now present a formal convergence guarantee for the gradient descent iterates.
\begin{theorem}
\label{Theorem: convergence analysis of multi circuit}
Suppose that the initialization $(\mU_1^0,\dots,\mU_N^0)$ satisfies $g(\mU_1^0,\dots,\mU_N^0) - g(\mU_1^\star,\dots,\mU_N^\star)\leq \frac{\Delta_1}{2}$. where $\Delta_1 = E_s - E_0>0$ denotes the spectral gap. Let $\{(\mU_1^t,\dots,\mU_N^t)\}_{t\geq 1}$ be the sequence generated by Riemannian gradient descent applied to \eqref{iteration of RGD on the Stiefel N layer} with step size $\mu \leq \frac{1}{(8N+1)\|\mH\|}$. Then, the iterates converge linearly to $(\mU_1^\star,\dots,\mU_N^\star)$ in the sense that
\begin{eqnarray}
\label{descent direction of multi layer quantum circuit final conclusion1 main paper}
 g(\mU_1^{t},\ldots,\mU_N^{t})- g(\mU_1^\star,\dots,\mU_N^\star)\leq\bigg(1 -\frac{\Delta_1^2\mu}{32\|\mH\|}\bigg)^t\big(g(\mU_1^0,\dots,\mU_N^0) - g(\mU_1^\star,\dots,\mU_N^\star) \big).
\end{eqnarray}
\end{theorem}
The proof is provided in Appendix~\ref{Proof of convergence analysis of multi circuit appendix}. We note that by choosing $\mu = \frac{1}{(8N+1)\|\mH\|}$, the convergence rate $1 - \frac{\Delta_1^2}{32(8N+1)\|\mH\|^2}$ depends only polynomially on $N$. In general nonconvex optimization problems, multi-layer parameterizations often suffer from exploding or vanishing gradients during training \cite{bengio1994learning,pascanu2013difficulty,le2015simple,arjovsky2016unitary,hanin2018neural,han2022optimal}, primarily due to unstable propagation across layers. In contrast, the Riemannian optimization on the unitary group preserves norm structure at each layer, which prevents exponential growth or decay of the effective gradient. As a result, the convergence behavior of \eqref{iteration of RGD on the Stiefel N layer} exhibits only linear dependence on the depth $N$, rather than the exponential instability typically observed in general deep nonconvex models. This phenomenon is consistent with recent results on structured tensor-factor learning \cite{qin2024guaranteed,qin2025robust,qin2025computational,qin2025scalable}. In addition, the analysis of the vanishing Riemannian gradient for \eqref{loss function of N layer} follows the same general strategy as in \Cref{subsec One-Layer Quantum Circuit}. The main difference is that we now consider $\ket{\vphi_1} = \mU_1\cdots\mU_N \ket{\vphi_0} = \sum_{k=0}^{D-1} c_k \ket{\vpsi_k}$, where the coefficients satisfy $\sum_{k=0}^{D-1} |c_k|^2 = 1$. In other words, the single-unitary parametrization is replaced by a product of unitary matrices. Consequently, the Riemannian gradient, measured in the Frobenius norm, can become exponentially small when the initialization condition is not satisfied.

The above convergence result is established under the ansatz-free product-unitary formulation, where each layer $\mU_h\in\calU(D)$ is treated as an independent optimization variable. It therefore characterizes an ideal geometric benchmark for the best achievable optimization behavior. In practical variational quantum circuits, however, each unitary layer is typically constrained by a finite-dimensional parametrization, which introduces an additional expressivity bottleneck absent in the ansatz-free setting. To make this connection explicit, we consider the layered variational ansatz
\begin{eqnarray}
\label{variational ansatz main paper}
\mU(\vtheta)=\mU_1(\theta_1)\cdots \mU_N(\theta_N),
\end{eqnarray}
where each layer $\mU_h(\theta_h)$ is a smooth parametrization of an ansatz manifold $\mathcal A\subseteq \mathrm{SU}(D)$ with parameter $\theta_h\in\mathbb R^p$. Here, $\mathrm{SU}(D):=\{\mU\in\mathbb C^{D\times D}:\mU^\dagger\mU=\mId,\ \det(\mU)=1\}$ denotes the special unitary group. We restrict the variational ansatz to $\mathrm{SU}(D)$ rather than the full unitary group $\calU(D)$, since global phase factors are physically irrelevant in quantum state evolution. The quantity $p=\dim(\mathcal A)$ denotes the intrinsic dimension of the ansatz manifold. Suppose the target unitary family lies in a structured subset
\begin{eqnarray}
\mathcal U_{\mathrm{target}}\subseteq \mathrm{SU}(D),
\end{eqnarray}
with intrinsic dimension $d_{\mathrm{target}}:=\dim(\mathcal U_{\mathrm{target}})$. Then, by the dimension-counting argument in Appendix~\ref{proof of dimension of multilayers}, a necessary condition for the variational ansatz to represent the target family is
\begin{eqnarray}
\label{necessary condition of required layers main paper}
N\ge \left\lceil \frac{d_{\mathrm{target}}}{p}\right\rceil.
\end{eqnarray}
From the above discussions, we can identify several interrelated mechanisms underlying the emergence of the BP phenomenon, which can be interpreted from a unified geometric perspective.
\begin{itemize}

\item \textbf{Depth--expressivity tradeoff.} When the circuit depth $N$ increases, the expressivity of the variational ansatz increases accordingly. However, the corresponding convergence rate $1 - \frac{\Delta_1^2}{32(8N+1)\|\mH\|^2}$ approaches $1$, indicating increasingly slow convergence and an effectively flatter optimization landscape. This is consistent with the observation in \cite{holmes2022connecting} that highly expressive ansatzes tend to exhibit flatter landscapes and reduced trainability.

\item \textbf{Exponential complexity of universal expressivity.} In the extreme case where the variational family is required to represent the full unitary group, i.e., $\mathcal U_{\mathrm{target}} = \mathrm{SU}(D)$, we have $d_{\mathrm{target}} = D^2 - 1$. For an $n$-qubit system with $D=2^n$, this yields $d_{\mathrm{target}} = 4^n - 1$. In this regime, the required circuit depth satisfies $N \ge \lceil \frac{4^n - 1}{p} \rceil$, which grows exponentially unless the per-layer expressivity $p$ scales accordingly. As a consequence, the convergence rate becomes exponentially suppressed in $n$, reflecting the intrinsic geometric complexity of the full unitary manifold.

\item \textbf{Per-layer expressivity bottleneck.} For standard hardware-efficient ansatzes, each layer consists of single-qubit rotations followed by fixed entangling gates such as CNOT or Controlled-Z. Since only the local rotations contribute trainable parameters, the intrinsic dimension per layer is $p=3n$. Consequently, $N \ge \lceil \frac{d_{\mathrm{target}}}{3n} \rceil$. This indicates that increasing single-layer expressivity can reduce required depth; however, for exponentially large target families, a polynomially expressive layer is insufficient to avoid large depth scaling.

\end{itemize}
Overall, from a convergence-theoretic perspective, these observations suggest that the BP phenomenon arises from a fundamental interplay between $(i)$ the exponential geometry of the Hilbert space, $(ii)$ the intrinsic dimension of the target unitary family, $(iii)$ the per-layer expressivity of the variational ansatz, and $(iv)$ the circuit depth required to bridge these mismatches.

\paragraph{Global Landscape Structure.} In addition to the local behavior discussed above, the global geometry of \eqref{loss function of N layer} exhibits significantly more intricate structure. In contrast to the single-unitary objective, compositions of multiple (in particular, more than two) matrix factors can give rise to higher-order saddle points and potentially spurious local minima far from the global optimum \cite{vidal2022optimization}. As a result, a global landscape characterization analogous to the single-unitary case is generally unavailable, and local first- and second-order information (gradient and Hessian) may be insufficient to reliably avoid such non-global critical points. An alternative that can locally restore favorable global structure for each subproblem is block-wise optimization, which has been explored in parametrized quantum circuit learning \cite{ding2024random,raymond2025optimizing,lai2026interpolation}. Concretely, when optimizing $g(\mU_1,\ldots,\mU_N)$ in \eqref{loss function of N layer}, we may update a single block $\mU_{h}$ while keeping the remaining factors fixed. The objective then reduces to a single-unitary form,
\begin{eqnarray}
\label{one-layer format of N layers quantum circuits}
 \ol{g}(\mU_{h}) = \langle \wt{\vphi}_0 |  \mU_{h}^\dagger \wt{\mH} \mU_{h}  | \wt{\vphi}_0 \rangle,
\end{eqnarray}
where $| \wt{\vphi}_0 \rangle = \mU_{h+1} \cdots \mU_N | \vphi_0 \rangle$ and $\wt{\mH} = \mU_{h-1}^\dagger \cdots \mU_{h}^\dagger \mH \mU_{h}\cdots \mU_{h-1}$ depend on the current values of the remaining blocks. This shows that, when all other factors are fixed, the optimization over $\mU_{h}$ takes the same form as the single-unitary variational problem.  Consequently, the global landscape characterization established in Theorem~\ref{Theorem: global optimization of one-layer quantum circuit} applies to each block subproblem: every critical point of the corresponding layerwise objective is either a global minimizer of that subproblem or a strict saddle point, implying the absence of spurious local minima at the block level. When the optimization is restricted to the scalar phase parameter $\theta_h$ in $\mU(\theta_h)$, rather than the full unitary block, \cite{wiedmann2025convergence} further shows that the Riemannian gradient vanishes only at global optima and strict saddle points, and that the latter are avoided almost surely by gradient descent under Lipschitz continuity of the gradient. Nevertheless, the existence of a global minimizer for each individual block subproblem does not imply that the composition of these minimizers yields a global minimizer of the full objective $g(\mU_1,\ldots,\mU_N)$. This is because the full objective is not jointly separable across the unitary factors: the optimal choice of $\mU_h$ depends on the configuration of all other blocks, and optimizing each factor independently in a sequential fashion may lead to a fixed point that is only locally optimal for the overall problem. Therefore, while the favorable landscape structure of each block subproblem provides a theoretical foundation for layerwise optimization, a carefully chosen initialization of the full parameter $(\mU_1,\ldots,\mU_N)$ remains essential for ensuring convergence to a globally optimal solution.

\section{Initialization Guarantee for Convergence}
\label{Sec:initialization requirement}

In the previous section, the convergence guarantee was established under the assumption that the initialization satisfies
\begin{eqnarray}
\label{initialization requirement for target}
g(\mU_1^0,\dots,\mU_N^0) - g(\mU_1^\star,\dots,\mU_N^\star)\leq \frac{\Delta_1}{2}.
\end{eqnarray}
While this condition ensures convergence, it is not a priori clear whether it is achievable in practice. This raises the question of existence: {\it{can one construct an initialization scheme that satisfies the above requirement with high probability?}}

In this section, we answer this question in the affirmative by establishing the existence of a feasible initialization scheme. Specifically, we show that such an initialization can be constructed through a practical variational quantum circuit, where a parameterized Pauli-rotation circuit with small random angles satisfies the required condition with high probability. The randomness is controlled by a variance parameter $\sigma^2$, and the guarantee holds provided that both the perturbation magnitude and the choice of the reference state $|\vphi_0\rangle$ are appropriately specified. Concretely, we consider the following initialization scheme:
\begin{eqnarray}
\label{initialization design}
\mU_h^0 := \mU_h^0(\theta_h)  = e^{-i\theta_h\mP_h} = \cos(\theta_h)\mId - i\sin(\theta_h) \mP_h,  \ h\in[N],
\end{eqnarray}
where $\mP_h$ is a Pauli matrix satisfying $\mP_h^2 = \mId$, and $\{\theta_h\}_{h=1}^N$ are i.i.d. random variables drawn from $\calN(0,\sigma^2)$.

To establish the desired initialization guarantee, it is necessary to quantify how close the initialized unitary matrices are to the optimal ones.
Specifically, we establish a high-probability upper bound on the initialization gap $g(\mU_1^0,\dots,\mU_N^0) - g(\mU_1^\star,\dots, \mU_N^\star)$ under the random initialization scheme in \eqref{initialization design}. The following theorem formalizes this bound under the initialization scheme described in \eqref{initialization design}.
\begin{theorem}
\label{Theorem: initialization design of multi circuit}
Let $g(\mU_1,\ldots,\mU_N)$ be the objective function defined in \eqref{loss function of N layer}, and consider the initialization scheme in \eqref{initialization design}. Then, with probability at least $1-\delta$ over the random initialization $\{\theta_h\}_{h=1}^N$, the initialization error satisfies
\begin{eqnarray}
\label{upper bound of gap difference main result}
g(\mU_1^0,\dots, \mU_N^0) - g(\mU_1^\star,\dots, \mU_N^\star) &\!\!\!\!=\!\!\!\!& g(\mU_1^0(\theta_1),\dots, \mU_N^0(\theta_N)) - g(\mU_1^\star,\dots, \mU_N^\star) \nonumber\\
&\!\!\!\!\leq\!\!\!\!& \bigg(\frac{1+e^{-2\sigma^2}}{2}\bigg)^N \big(\langle\vphi_0| \mH |\vphi_0 \rangle - g(\mU_1^\star,\dots, \mU_N^\star) \big)\nonumber\\
&\!\!\!\!\!\!\!\!& + \bigg(1 - \bigg(\frac{1+e^{-2\sigma^2}}{2}\bigg)^N\bigg) \big(\|\mH\| - g(\mU_1^\star,\dots, \mU_N^\star)\big)\nonumber\\
&\!\!\!\!\!\!\!\!& + 2\|\mH\|\sigma\sqrt{2N\log(2/\delta) }.
\end{eqnarray}
\end{theorem}
The proof is provided in {Appendix}~\ref{Proof of initialization requirement appendix}. From Theorem~\ref{Theorem: initialization design of multi circuit}, it follows that, in order to guarantee $g(\mU_1^0,\dots,\mU_N^0) - g(\mU_1^\star,\dots,\mU_N^\star)\leq \frac{\Delta_1}{2}$, each of the three terms in the upper bound \eqref{upper bound of gap difference main result} must be controlled appropriately. Specifically, we can require that
\begin{eqnarray}
\label{requirement of three terms}
\begin{cases}
\big(\frac{1+e^{-2\sigma^2}}{2}\big)^N (\langle\vphi_0| \mH |\vphi_0 \rangle - g(\mU_1^\star,\dots, \mU_N^\star) ) \leq c_1\Delta_1,\\
\big(1 - \big(\frac{1+e^{-2\sigma^2}}{2}\big)^N\big) (\|\mH\| - g(\mU_1^\star,\dots, \mU_N^\star)) \leq c_2\Delta_1,\\
2\|\mH\|\sigma\sqrt{2N\log(2/\delta) } \leq c_3\Delta_1,
\end{cases}
\end{eqnarray}
where $c_1$, $c_2$, and $c_3$ are positive constants satisfying $c_1+c_2+c_3 =\frac{1}{2}$. To simplify the analysis, we consider the regime where $\sigma^2$ is small.
In this case, we have the first-order approximation $\frac{1+e^{-2\sigma^2}}{2} \approx e^{-\sigma^2}$, and consequently $\big(\frac{1+e^{-2\sigma^2}}{2}\big)^N \approx e^{-N\sigma^2}$. Under this approximation, the three conditions in \eqref{requirement of three terms}
yield the following requirement on $\sigma^2$:
\begin{eqnarray}
\label{requirement of sigma}
\sigma^2\in &\!\!\!\!\!\!\!\!&\bigg[\frac{1}{N}\log\frac{\langle\vphi_0|\mH|\vphi_0\rangle - g(\mU_1^\star,\dots,\mU_N^\star)}
{c_1\Delta_1}, \nonumber\\
&\!\!\!\!\!\!\!\!& \min\bigg\{\frac{1}{N}\log\frac{\|\mH\| - g(\mU_1^\star,\dots, \mU_N^\star)}{\|\mH\| - g(\mU_1^\star,\dots, \mU_N^\star) - c_2\Delta_1},\
\frac{c_3^2\Delta_1^2}{8N\|\mH\|^2\log(2/\delta)}
\bigg\} \bigg].
\end{eqnarray}
Based on \eqref{requirement of sigma}, we make the following observations regarding the initialization strategy.
\begin{itemize}
\item \textbf{Feasibility condition.} The feasibility of the range \eqref{requirement of sigma} requires the lower bound not to exceed both upper bounds, which imposes conditions on both the reference state $|\vphi_0\rangle$ and the confidence parameter $\delta$. Specifically, the initial energy gap $\langle\vphi_0|\mH|\vphi_0\rangle - g(\mU_1^\star,\dots,\mU_N^\star)$ must be sufficiently small, meaning that $|\vphi_0\rangle$ must already be a reasonable approximation to the ground state in terms of energy. At the same time, the confidence parameter $\delta$ must not be chosen too small, as an overly stringent confidence requirement shrinks the upper bound and may render the feasible range empty. Together, these conditions underscore that the existence of a valid initialization variance $\sigma^2$ depends jointly on the quality of the reference state and the choice of $\delta$.

\item \textbf{Dependence on initialization quality.} The lower bound on $\sigma^2$ is determined by the quality of the reference state $|\vphi_0\rangle$: a smaller initial energy gap $\langle\vphi_0|\mH|\vphi_0\rangle - g(\mU_1^\star,\dots,\mU_N^\star)$ allows a smaller $\sigma^2$, meaning the unitary matrices can be initialized closer to the identity. This is consistent with the intuition that deeper circuits require initialization closer to the identity, a principle also underlying the identity block initialization strategy of \cite{grant2019initialization}.

\item \textbf{Scaling behavior.} Both bounds scale as $O(1/N)$, implying $\sigma = O(1/\sqrt{N})$, so that as the number of unitary factors $N$ increases, the perturbation must be correspondingly reduced.
\end{itemize}
Together, these observations highlight that a good initialization requires both the reference state $|\vphi_0\rangle$ to be close to the ground state eigenvectors $\{|\vpsi_k\rangle\}_{k=0}^{s-1}$ and the initial unitary matrices to be sufficiently close to the identity.

We conclude this section with a remark on the practical limitations of the proposed initialization scheme. The valid range of $\sigma^2$ depends on problem-specific quantities such as the spectral gap $\Delta_1$, the operator norm $\|\mH\|$, and the ground-state energy $g(\mU_1^\star,\dots,\mU_N^\star)$, which are generally not accessible in practice. In addition, the choice of Pauli operators $\{\mP_h\}$ and the reference state $|\vphi_0\rangle$ may not be directly implementable without prior knowledge of the underlying problem structure. Despite these limitations, the purpose of this section is to establish a theoretical guarantee of feasibility: there exists an initialization scheme that satisfies $g(\mU_1^0,\dots,\mU_N^0) - g(\mU_1^\star,\dots,\mU_N^\star)\leq \frac{\Delta_1}{2}$ with high probability. This ensures that the optimization trajectory of Riemannian gradient descent remains in a locally well-behaved region, where the linear convergence result derived in the previous section applies.

\section{Noise Robustness of Variational Quantum Eigensolver}

\subsection{Uniform Measurement Allocation}

In the preceding sections, the analysis has been carried out under the assumption of exact measurements. In practice, however, the expectation values of observables can only be estimated from a finite number of measurement shots, introducing statistical noise. In this section, we extend our analysis to the noisy setting by considering a stochastic approximation of the objective function.

To this end, we decompose the Hamiltonian as $\mH = \sum_{k=1}^{L}\alpha_k \mP_k$, where $\{\mP_k\}_{k=1}^{L}$ are Pauli operators and $\alpha_k \in \R$ are the corresponding
coefficients. Note that such a decomposition always exists for any Hermitian matrix and satisfies $L \leq 4^n$. Substituting this into \eqref{loss function of N layer}, the objective function can be written as
\begin{eqnarray}
\label{loss function of N layer another form}
g(\mU_1,\ldots,\mU_N) &\!\!\!\!=\!\!\!\!& \sum_{k=1}^{L}\alpha_k \langle \vphi_0 | \mU_N^\dagger \cdots \mU_1^\dagger \mP_k \mU_1 \cdots \mU_N | \vphi_0 \rangle\nonumber\\
&\!\!\!\!=\!\!\!\!& \sum_{k=1}^{L}\alpha_k (p_k^+ - p_k^-),
\end{eqnarray}
where $|\vphi_N\rangle := \mU_1\cdots\mU_N|\vphi_0\rangle$ denotes the
output state of the circuit. It follows directly that
\begin{eqnarray}
p_k^+ - p_k^- =  \frac{1}{2}\langle\vphi_N|(\mId + \mP_k)|\vphi_N\rangle - \frac{1}{2}\langle\vphi_N|(\mId - \mP_k)|\vphi_N\rangle  = \langle\vphi_N|\mP_k|\vphi_N\rangle,
\end{eqnarray}
which establishes the second equality in \eqref{loss function of N layer another form}. Note that $p_k^+$ and $p_k^-$ denote the probabilities of obtaining the $+1$ and $-1$ measurement outcomes of $\mP_k$, respectively, so that $p_k^+ + p_k^- = 1$. In practice they are unknown and must be estimated from a finite number of measurement shots. Specifically, given $M$ independent measurement shots of $\mP_k$, let $f_k^+$ and $f_k^-$ denote the number of times the outcomes $+1$ and $-1$ are observed, respectively, with $f_k^+ + f_k^- = M$. The corresponding empirical probabilities are given by $\hat{p}_k^+ = \frac{f_k^+}{M}$ and $\hat{p}_k^- = \frac{f_k^-}{M}$, which yield a noisy estimate of the objective function:
\begin{eqnarray}
\label{noisy objective}
\wh{g}(\mU_1,\ldots,\mU_N) = \sum_{k=1}^{L}\alpha_k(\hat{p}_k^+ - \hat{p}_k^-)
= g(\mU_1,\ldots,\mU_N) + \eta,
\end{eqnarray}
where $\eta = \sum_{k=1}^{L}\alpha_k \big(  \hat{p}_k^+ - \hat{p}_k^- - (p_k^+ - p_k^-)  \big)$ denotes the total measurement noise. In what follows, we analyze the impact of measurement noise on the convergence behavior and robustness of Riemannian gradient descent on the unitary group.
\begin{theorem}
\label{Theorem: robustness of noisy multi circuit}
Suppose that the initialization $(\mU_1^0,\dots,\mU_N^0)$ satisfies $g(\mU_1^0,\dots,\mU_N^0) - g(\mU_1^\star,\dots,\mU_N^\star)\leq \frac{\Delta_1}{2}$. where $\Delta_1 = E_s - E_0>0$ denotes the spectral gap. Let $\{(\mU_1^t,\dots,\mU_N^t)\}_{t\geq 1}$ be the sequence generated by Riemannian gradient descent applied to \eqref{noisy objective} with step size $\mu \leq \frac{1}{(8N+1)\|\mH\|}$. Then, with probability at least $1-\gamma$, the iterates exhibit a linear convergence rate toward a neighborhood of $(\mU_1^\star,\dots,\mU_N^\star)$, where the size of the neighborhood is determined by the measurement noise level. Specifically, given $M$ independent measurement shots of each Pauli operator $\mP_k$, the following bound holds:
\begin{eqnarray}
\label{noisy objective expansion final version main paper}
&\!\!\!\!\!\!\!\!&\wh{g}(\mU_1^{t},\ldots,\mU_N^{t}) - g(\mU_1^\star,\ldots,\mU_N^\star)\nonumber\\
&\!\!\!\! \leq \!\!\!\!& \sqrt{2\log(1/\gamma)}\sqrt{\frac{\sum_{k=1}^L \alpha_k^2 }{M}} + \bigg(1 -\frac{\Delta_1^2\mu}{32\|\mH\|}\bigg)^{t}\big(g(\mU_1^0,\dots,\mU_N^0) - g(\mU_1^\star,\dots,\mU_N^\star) \big).
\end{eqnarray}
\end{theorem}
The proof is provided in {Appendix}~\ref{Proof of error bound appendix}. \Cref{Theorem: robustness of noisy multi circuit} shows that the convergence behavior under noisy measurements consists of two components. The first term represents the irreducible statistical error due to finite measurements, which is independent of the number of iterations $t$ and decreases at a rate of $O(1/\sqrt{M})$ as the number of measurement shots $M$ increases. The second term is the optimization error, which decays exponentially at the same linear rate $1 - \frac{\Delta_1^2\mu}{32\|\mH\|}$ as in the noiseless case (\Cref{Theorem: convergence analysis of multi circuit}). Consequently, the iterates converge to a neighborhood of the global minimizer whose size is controlled by the noise level. In particular, by taking $M = O\big(\frac{\sum_{k=1}^L\alpha_k^2}{\epsilon^2}\big)$ measurement shots, the statistical error can be made smaller than any target precision $\epsilon > 0$.

We further extend the conclusion of \Cref{Theorem: robustness of noisy multi circuit} to the variational objective in \eqref{objective function in the introduction parameterized}. Consider the noisy variational objective:
\begin{eqnarray}
\label{noisy objective fixed ansatz}
\wh{G}(\theta_1,\ldots,\theta_N) = G(\theta_1,\ldots,\theta_N) + \eta,
\end{eqnarray}
where $\eta$ captures the measurement noise induced by finite-shot measurements.  Since the noise term $\eta$ depends only on the measurement statistics---specifically, on the number of shots $M$ and the Pauli coefficients $\{\alpha_k\}$---and is independent of the circuit parameterization $(\theta_1,\ldots,\theta_N)$, the concentration analysis of Appendix~\ref{Proof of error bound appendix} extends directly to the parameterized setting.
\begin{cor}
\label{Corollary: noisy objective error bound}
Let $(\theta_1^\star,\ldots,\theta_N^\star)$ be an optimal solution of the noiseless objective $G(\theta_1,\ldots,\theta_N)$, and let $(\wh{\theta}_1,\ldots,\wh{\theta}_N)$ be the output of an optimization algorithm applied to the noisy objective $\wh{G}$ in the steady-state regime, where the optimization error is sufficiently small compared with the statistical noise level. Then, given $M$ independent measurement shots of each Pauli operator $\mP_k$, with probability at least $1-\gamma$, we have
\begin{eqnarray}
\label{noisy objective expansion final version main paper conclusion}
\wh{G}(\wh\theta_1,\ldots,\wh\theta_N)
-
G(\theta_1^\star,\ldots,\theta_N^\star)
=
O\bigg(\sqrt{\sum_{k=1}^L\frac{\log(1/\gamma) \alpha_k^2}{M}}\bigg).
\end{eqnarray}
\end{cor}
This result highlights that, once the optimization reaches the noise-dominated regime, the dominant limitation on accuracy arises from finite-shot measurement statistics rather than the optimization landscape. Consequently, the achievable estimation accuracy is fundamentally limited by the measurement budget $M$.

\subsection{Non-Uniform Measurement Allocation}

In the previous analysis, we assumed a uniform measurement budget, i.e., each Pauli operator $\mP_k$ is measured using the same number of measurement shots $M$. While this simplifies the presentation, such a uniform allocation is generally suboptimal when the total measurement budget is fixed. To better utilize the available resources, we now consider a heterogeneous allocation strategy. Specifically, given a total measurement budget, we assign a potentially different number of measurement shots $M_k$ to each Pauli operator $\mP_k$, for $k=1,\ldots,L$, thereby enabling an adaptive measurement strategy that accounts for the heterogeneous contributions of different Hamiltonian terms. Under this more general setting, we establish the following result, which characterizes the local geometry of the Riemannian gradient descent algorithm with non-uniform measurement allocation.
\begin{theorem}
\label{lemma: robustness of noisy multi circuit}
Under the same conditions as in \Cref{Theorem: robustness of noisy multi circuit}, suppose that for each Pauli operator $\mP_k$, $M_k$ independent measurements are performed. Then, given $M_k$ independent measurement shots of each Pauli operator $\mP_k$, with probability at least $1-\gamma$, the iterates satisfy
\begin{eqnarray}
\label{noisy objective expansion final version M main paper}
&\!\!\!\!\!\!\!\!&\wh{g}(\mU_1^{t},\ldots,\mU_N^{t}) - g(\mU_1^\star,\ldots,\mU_N^\star)\nonumber\\
&\!\!\!\! \leq \!\!\!\!& \sqrt{2\log(1/\gamma)}\sqrt{\sum_{k=1}^L \frac{\alpha_k^2}{M_k} } + \bigg(1 -\frac{\Delta_1^2\mu}{32\|\mH\|}\bigg)^{t}\big(g(\mU_1^0,\dots,\mU_N^0) - g(\mU_1^\star,\dots,\mU_N^\star) \big).
\end{eqnarray}
\end{theorem}
The detailed proof is deferred to Appendix~\ref{Proof of optimal allocation}. This result generalizes \Cref{Theorem: robustness of noisy multi circuit} to the non-uniform sampling regime. In particular, the statistical error term now explicitly depends on the allocation $\{M_k\}_{k=1}^L$, revealing how different measurement strategies affect the overall estimation accuracy. This formulation naturally enables an optimization of the measurement allocation under a fixed total budget, leading to an improved trade-off between statistical error and sampling cost. Specifically, we formulate the problem of minimizing the total error under a fixed measurement budget $M_{\text{tot}}$ as the following constrained optimization problem:
\begin{eqnarray}
\label{minimization of total error under constraints}
\min_{M_1,\dots,M_L} \sum_{k=1}^L \frac{\alpha_k^2}{M_k}, \quad \text{s.t. } \sum_{k=1}^L M_k = M_{\text{tot}}.
\end{eqnarray}
Based on the derivation in Appendix~\ref{Proof of optimal allocation closed form}, the optimal solution to \eqref{minimization of total error under constraints} is given by
\begin{eqnarray}
\label{closed-form of constrained optimization main paper}
\wh{M}_k = \frac{|\alpha_k|}{\sum_{j=1}^L |\alpha_j|} M_{\text{tot}}, \ \ k=1,\dots,L.
\end{eqnarray}
Substituting this allocation into the error formula in \eqref{minimization of total error under constraints} gives $\sum_{k=1}^L \frac{\alpha_k^2}{\wh{M}_k} = \frac{(\sum_{k=1}^L |\alpha_k|)^{2}}{M_{\text{tot}}}$. Finally, by Cauchy--Schwarz inequality, we have
\begin{eqnarray}
\label{error formula 2}
\sum_{k=1}^L \frac{\alpha_k^2}{\wh{M}_k} = \frac{(\sum_{k=1}^L |\alpha_k|)^{2}}{M_{\text{tot}}} \leq \frac{L \sum_{k=1}^L \alpha_k^{2}}{M_{\text{tot}}} = \frac{ \sum_{k=1}^L \alpha_k^{2}}{M},
\end{eqnarray}
where the inequality becomes equality if and only if $\alpha_k$ are identical for all $k$. Here $M = M_{\text{tot}}/L$ and the right-hand side corresponds to the uniform allocation. This establishes that, under a fixed total measurement budget, the optimal non-uniform allocation strictly reduces the statistical error compared to uniform allocation whenever the coefficients $\{\alpha_k\}$ are not all equal. Finally, we clarify the distinction between the present setting and the optimal measurement allocation studied in \cite{qin2026optimal}. The work \cite{qin2026optimal} considers a joint optimization over both the number of Pauli measurement settings and the number of repetitions per setting. In contrast, in our setting the measurement operators are fixed, and we focus solely on the allocation of repeated measurements across a given set of Pauli operators.

Motivated by this observation, we extend the uniform-shot error analysis in Corollary~\ref{Corollary: noisy objective error bound} to the general non-uniform measurement allocation case.
\begin{cor}
\label{Corollary: noisy objective error bound nonuniform}
Let $(\theta_1^\star,\ldots,\theta_N^\star)$ be an optimal solution of the noiseless objective $G(\theta_1,\ldots,\theta_N)$, and let $(\wh{\theta}_1,\ldots,\wh{\theta}_N)$ be the output of an optimization algorithm applied to the noisy objective $\wh{G}$ in the steady-state regime, where the optimization error is sufficiently small compared with the statistical noise level. Then, given $M_k$ independent measurement shots of each Pauli operator $\mP_k$, with probability at least $1-\gamma$, we have
\begin{eqnarray}
\label{noisy objective expansion final version main paper nonuniform conclusion}
\wh{G}(\wh\theta_1,\ldots,\wh\theta_N)
-
G(\theta_1^\star,\ldots,\theta_N^\star)
=
O\bigg(\sqrt{\sum_{k=1}^L \frac{\log(1/\gamma)\alpha_k^2}{M_k}}\bigg).
\end{eqnarray}
\end{cor}

\section{Simulation}

In this section, we conduct numerical experiments to corroborate the theoretical results developed in the previous sections. We implement the RGD algorithm in \eqref{iteration of RGD on the Stiefel N layer} with step size $\mu = \mu_0 / N$. In the initialization scheme \eqref{initialization design}, each Pauli operator is sampled uniformly at random from $\{\mId,\mX,\mY,\mZ\}^{\otimes n}$, and the phase parameters $\{\theta_h\}$ are independently drawn from $\mathcal{N}(0,\sigma^2)$. We set the Hilbert space dimension to $D=2^n$ and choose the reference state as $|\vphi_0\rangle = [1,0,\ldots,0]^\top$. All reported results are averaged over 100 independent Monte Carlo trials.

\begin{figure}[!ht]
\centering
\subfigure[]{
\begin{minipage}[t]{0.305\textwidth}
\centering
\includegraphics[width=5.5cm]{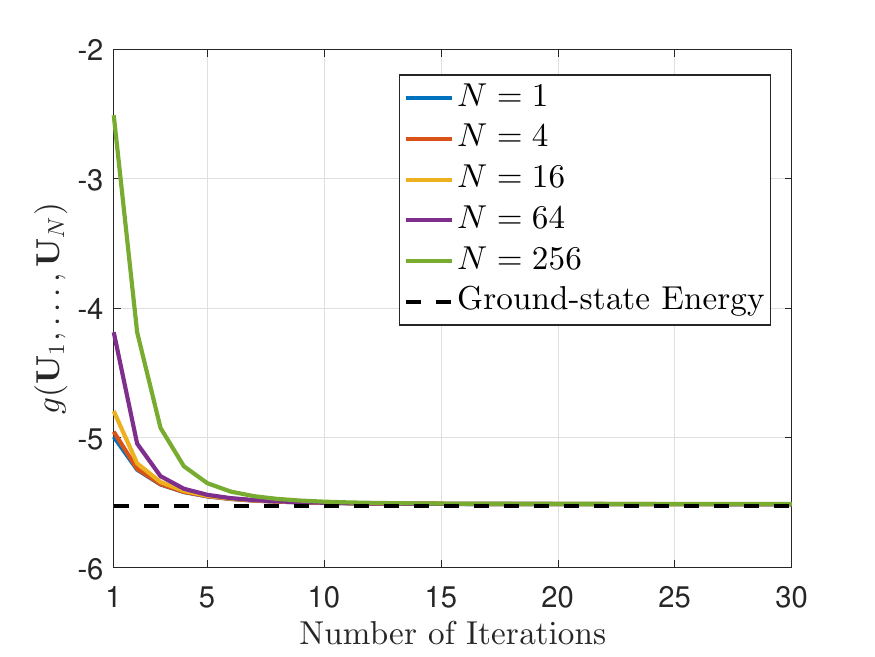}
\end{minipage}
\label{convergence analysis of N}
}
\subfigure[]{
\begin{minipage}[t]{0.305\textwidth}
\centering
\includegraphics[width=5.5cm]{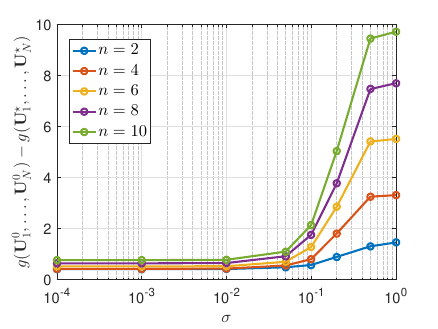}
\end{minipage}
\label{Initialization requirement}
}
\subfigure[]{
\begin{minipage}[t]{0.305\textwidth}
\centering
\includegraphics[width=5.5cm]{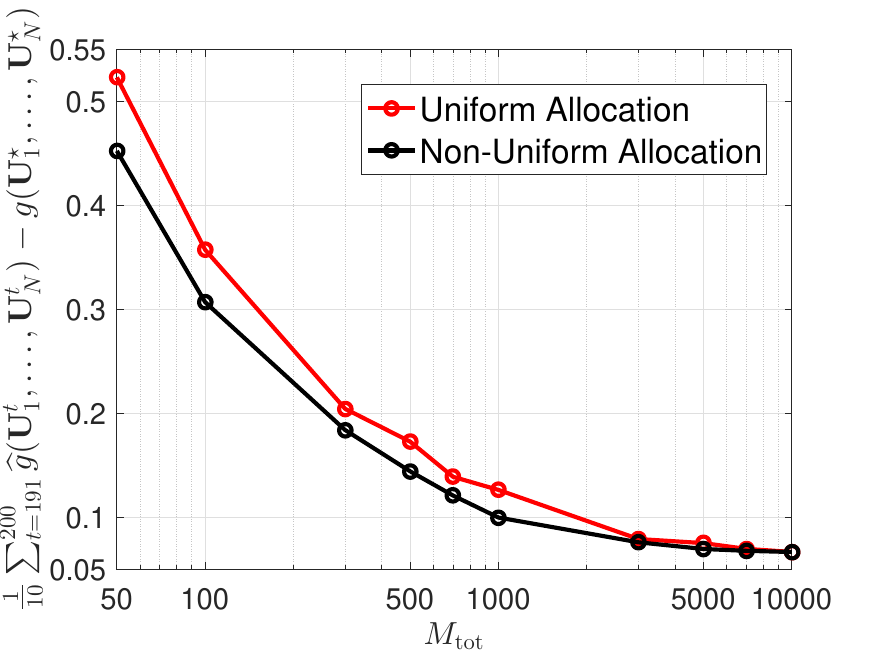}
\end{minipage}
\label{Recovery error bound}
}
\caption{(a) Convergence behavior of RGD for varying $N$ with $\mu_0=0.2$, $n=6$, and $\sigma=0.05$; (b) Initialization error for different $\sigma$ and $n$ with $N=16$; (c) Comparison between uniform and non-uniform measurement allocation under varying total budget $M_{\textup{tot}}$ with $\mu_0=0.2$, $n=6$, $N=8$, and $\sigma=0.05$.}
\end{figure}

In the first experiment, we examine the dependence of convergence on the number of layers $N$ under the transverse-field Ising Hamiltonian \cite{pfeuty1971ising} $\mH = - 0.5 \sum_{k=1}^{n} \mX_k - \sum_{k=1}^{n-1} \mZ_k \mZ_{k+1}$, where $\mX_k$ and $\mZ_k$ denote standard Pauli operators acting on the $k$-th qubit. As shown in \Cref{convergence analysis of N}, increasing $N$ leads to a slower convergence rate. This behavior is consistent with \Cref{Theorem: convergence analysis of single circuit}, which shows that the convergence rate deteriorates as the number of layers increases. Moreover, the initialization error grows with $N$, in agreement with \Cref{Theorem: initialization design of multi circuit}, indicating that deeper circuits lead to a larger  initialization error.

In the second experiment, we consider the same setup as in the first experiment, and study the effect of the initialization scale $\sigma$ on the initial recovery error under different system sizes $n$. The results in \Cref{Initialization requirement} show that smaller $\sigma$ yields a smaller initial recovery error, indicating that initializations closer to the identity lead to improved performance. This is consistent with \Cref{Theorem: initialization design of multi circuit}, which bounds the initialization error in terms of the distance to the identity. As $\sigma$ decreases, the initialization concentrates around the identity, resulting in a smaller initial objective gap. In addition, the initial recovery error increases with the system size $n$. This reflects the increasing dimensionality of the unitary manifold, under which a randomly initialized circuit becomes more dispersed and less aligned with the identity, resulting in a larger initial objective value.

In the last experiment, we compare uniform and non-uniform measurement allocation under a fixed total measurement budget $M_{\textup{tot}}$. We consider a random Pauli Hamiltonian of the form $\mH = \sum_{k=1}^{L} \alpha_k \mP_k$, where $L=24$, each $\mP_k$ is independently sampled from $\{\mId,\mX,\mY,\mZ\}^{\otimes 24}$, and $\alpha_k \sim \mathcal{N}(0,1)$ followed by normalization such that $\sum_{k=1}^{L}\alpha_k^2=1$. From \Cref{Recovery error bound}, the non-uniform allocation achieves a consistently lower estimation error compared with the uniform strategy, especially in the small-budget regime. This is in direct agreement with the analytical result in \eqref{error formula 2}, which shows that the uniform allocation yields a suboptimal benchmark value under the same total measurement budget. The improvement is attributed to the reallocation of measurement resources toward Pauli observables with larger coefficients $\{|\alpha_k|\}$, which directly reduces the dominant contributions in the error term $\sum_{k=1}^L \alpha_k^2 / M_k$. As $M_{\textup{tot}}$ increases, the performance gap gradually diminishes, since the contribution from each observable becomes uniformly suppressed.

\section{Conclusion}
\label{conclusion}

In this paper, we study the VQE from a geometric perspective, focusing on its optimization landscape, initialization conditions, and robustness under finite-shot measurements. First, we analyze the optimization landscape via Riemannian optimization on the unitary group. We consider both single-unitary and product-unitary settings. In the single-unitary case, we establish that RGD exhibits explicit linear convergence under appropriate initialization and that the objective satisfies the strict-saddle property, where all critical points are either global minimizers or strict saddles. In the product-unitary case, we show that RGD converges to a global minimizer under suitable initialization, where the convergence rate deteriorates polynomially with circuit depth, providing a geometric explanation of the BP phenomenon in terms of Hilbert-space geometry, target-unitary complexity, per-layer expressivity, and circuit depth. Second, we analyze initialization conditions for these convergence results and show that practical small-angle random Pauli-rotation circuits satisfy the required geometric conditions with high probability. Third, we study robustness under finite-shot measurements and show that RGD retains linear convergence up to a noise-dominated neighborhood of the global minimizer. Moreover, coefficient-adaptive measurement allocation achieves strictly lower estimation error than uniform sampling under a fixed measurement budget. Overall, these results provide a unified geometric foundation for understanding the trainability of VQE and may inform the design of more principled optimization strategies for variational quantum algorithms more broadly.

\section{Acknowledgments}
\label{sec: ack}

ZQ gratefully acknowledges support from the MICDE Research Scholars Program at the University of Michigan.

\newpage

\appendices

\section{Proof of \Cref{Theorem: convergence analysis of single circuit}}
\label{Proof of convergence analysis of single circuit appendix}

\begin{proof}

First, we bound the term $\|\text{grad}_{\mU}f(\mU_1) - \text{grad}_{\mU}f(\mU_2)\|_F$ for any $\mU_1,\mU_2\in\calU(D)$ as follows:
\begin{eqnarray}
\label{the L smooth term for gradient difference}
&\!\!\!\!\!\!\!\!&\|\text{grad}_{\mU}f(\mU_1) - \text{grad}_{\mU}f(\mU_2)\|_F\nonumber\\
&\!\!\!\!=\!\!\!\!&\|\mH(\mU_1 - \mU_2) \vrho_0 \|_F + \| \mU_1\text{sym}(\mU_1^\dagger\mH\mU_1\vrho_0) - \mU_2\text{sym}(\mU_2^\dagger\mH\mU_2\vrho_0)  \|_F\nonumber\\
&\!\!\!\!\leq\!\!\!\!& \|\mH\| \|\vrho_0 \| \|\mU_1 - \mU_2\|_F + \|(\mU_1 - \mU_2)\text{sym}(\mU_1^\dagger\mH\mU_1\vrho_0)  \|_F + \|\mU_2(\text{sym}(\mU_1^\dagger\mH\mU_1\vrho_0) - \text{sym}(\mU_2^\dagger\mH\mU_2\vrho_0) ) \|_F\nonumber\\
&\!\!\!\!\leq\!\!\!\!& 2\|\mH\| \|\mU_1 - \mU_2\|_F + \|\text{sym}(\mU_1^\dagger\mH\mU_1\vrho_0) - \text{sym}(\mU_2^\dagger\mH\mU_2\vrho_0) \|_F \nonumber\\
&\!\!\!\!\leq\!\!\!\!& 4\|\mH\| \|\mU_1 - \mU_2\|_F,
\end{eqnarray}
where the second inequality follows from $\|\vrho_0\| = \| |\vphi_0\rangle \langle \vphi_0|\| = \||\vphi_0\rangle \|_2^2 = 1 $, $\|\mU_2\|=1$ and $\|\text{sym}(\mU_1^\dagger\mH\mU_1\vrho_0) \|\leq \|\vrho_0\|\|\mU_1 \|^2\|\mH\|=\|\mH\|$. In addition, the last line uses $\|\text{sym}(\mU_1^\dagger\mH\mU_1\vrho_0) - \text{sym}(\mU_2^\dagger\mH\mU_2\vrho_0) \|_F = \frac{1}{2}\|(\mU_1 - \mU_2)^\dagger \mH\mU_1\vrho_0 +  \mU_2^\dagger\mH(\mU_1 - \mU_2 )\vrho_0 + \vrho_0 (\mU_1 - \mU_2)^\dagger \mH\mU_1 + \vrho_0 \mU_2^\dagger\mH(\mU_1 - \mU_2 )\|_F \leq 2\|\mH\| \|\mU_1 - \mU_2\|_F$.

Combining \Cref{lemma:Riemannian descent conclusion} with \eqref{the L smooth term for gradient difference} and taking $L = \frac{9\|\mH\|}{2}$, we have
\begin{eqnarray}
\label{L smooth condition for Riemmanin gradient}
f(\mU^{t+1} ) &\!\!\!\!=\!\!\!\!& f(\text{Retr}_{\mU}( - \mu \text{grad}_{\mU}f(\mU^t) ))\nonumber\\
&\!\!\!\!\leq\!\!\!\!& f(\mU^t ) - \mu(1 - L\mu)\|\text{grad}_{\mU}f(\mU^t) \|_F^2\nonumber\\
&\!\!\!\!\leq\!\!\!\!& f(\mU^t ) - \frac{\mu}{2}\|\text{grad}_{\mU}f(\mU^t) \|_F^2,
\end{eqnarray}
where $\mu\leq \frac{1}{2L}$ is used in the last line.

Next, we analyze a lower bound for $\|\text{grad}_{\mU}f(\mU) \|_F^2$. Using the decomposition $\|\text{grad}_{\mU}f(\mU) \|_F^2 = \|\mH \mU \vrho_0\|_F^2 - \|\mU\text{sym}(\mU^\dagger\mH\mU\vrho_0) \|_F^2$, we begin by expanding the first term $\|\mH \mU \vrho_0\|_F^2$ as follows:
\begin{eqnarray}
\label{expansion of first term in the Riemannian gradient}
\|\mH \mU \vrho_0\|_F^2 &\!\!\!\!=\!\!\!\!& \trace(\vrho_0 \mU^\dagger \mH^2 \mU \vrho_0 )\nonumber\\
&\!\!\!\!=\!\!\!\!& \trace(\vrho_0 \mU^\dagger \mH^2 \mU  )\nonumber\\
&\!\!\!\!=\!\!\!\!&\langle\vphi| \mH^2  |\vphi \rangle,
\end{eqnarray}
where we define $|\vphi \> = \mU |\vphi_0 \>$. In addition, by defining $\wt\mH := \mU^\dagger \mH \mU$, we obtain
\begin{eqnarray}
\label{expansion of second term in the Riemannian gradient}
\|\mU\text{sym}(\mU^\dagger\mH\mU\vrho_0) \|_F^2 &\!\!\!\!=\!\!\!\!& \|\text{sym}(\mU^\dagger\mH\mU\vrho_0) \|_F^2\nonumber\\
&\!\!\!\!=\!\!\!\!& \frac{1}{4}\| \wt\mH \vrho_0 + \vrho_0 \wt\mH \|_F^2\nonumber\\
&\!\!\!\!=\!\!\!\!& \frac{1}{2}\|\wt\mH \vrho_0\|_F^2 + \frac{1}{2}\Re{\trace(\wt\mH \vrho_0 \wt\mH \vrho_0 )}\nonumber\\
&\!\!\!\!=\!\!\!\!& \frac{1}{2}\trace(\vrho_0 \wt\mH^2 \vrho_0 ) + \frac{1}{2}|\<\vphi_0| \wt\mH | \vphi_0 \> |^2\nonumber\\
&\!\!\!\!=\!\!\!\!& \frac{1}{2}\langle\vphi| \mH^2  |\vphi \rangle + \frac{1}{2}(f(\mU))^2.
\end{eqnarray}

Based on \eqref{expansion of first term in the Riemannian gradient} and \eqref{expansion of second term in the Riemannian gradient}, we can further derive
\begin{eqnarray}
\label{expansion of the Riemannian gradient}
\|\text{grad}_{\mU}f(\mU) \|_F^2 =  \frac{\langle\vphi| \mH^2  |\vphi \rangle - (f(\mU))^2}{2} = \frac{1}{2}F_{|\vphi \rangle}(\mH).
\end{eqnarray}

Define $|\vphi \rangle = \mU|\vphi_0 \rangle = \sum_{k=0}^{D-1}c_k|\vpsi_k\>$ with $\sum_{k=0}^{D-1} |c_k|^2 =1$, and let $\mH = \sum_{k=0}^{D-1}E_k |\vpsi_k \rangle\langle \vpsi_k|$ be its spectral decomposition. Then, we have $\langle\vphi| \mH  |\vphi \rangle = \sum_{k=0}^{D-1}E_k |\<\vpsi_k|\vphi\> |^2 = \sum_{k=0}^{D-1}E_k |c_k|^2$ and $\mH^2 = \sum_{k=0}^{D-1}E_k^2 |\vpsi_k \rangle\langle \vpsi_k|$, $\langle\vphi| \mH^2  |\vphi \rangle = \sum_{k=0}^{D-1}E_k^2 |c_k|^2$. Hence, it follows that
\begin{eqnarray}
\label{expansion of the Riemannian gradient one specifc term}
F_{|\vphi \rangle}(\mH) &\!\!\!\!=\!\!\!\!& \sum_{k=0}^{D-1}E_k^2 |c_k|^2 - \bigg(\sum_{k=0}^{D-1}E_k |c_k|^2\bigg)^2\nonumber\\
&\!\!\!\!=\!\!\!\!& \sum_{k=0}^{D-1}E_k^2 |c_k|^2 - \sum_{k=0}^{D-1}\sum_{j=0}^{D-1}E_k E_j|c_k|^2|c_j|^2\nonumber\\
&\!\!\!\!=\!\!\!\!& \sum_{k=0}^{D-1}\sum_{j=0}^{D-1}E_k^2 |c_k|^2 |c_j|^2 - \sum_{k=0}^{D-1}\sum_{j=0}^{D-1}E_k E_j|c_k|^2|c_j|^2\nonumber\\
&\!\!\!\!=\!\!\!\!& \sum_{k=0}^{D-1}\sum_{j=0}^{D-1}E_k(E_k - E_j)|c_k|^2 |c_j|^2\nonumber\\
&\!\!\!\!=\!\!\!\!& \frac{1}{2}\sum_{k=0}^{D-1}\sum_{j=0}^{D-1}E_k(E_k - E_j)|c_k|^2 |c_j|^2 + \frac{1}{2}\sum_{k=0}^{D-1}\sum_{j=0}^{D-1}E_j(E_j - E_k)|c_k|^2 |c_j|^2\nonumber\\
&\!\!\!\!=\!\!\!\!& \frac{1}{2}\sum_{k=0}^{D-1}\sum_{j=0}^{D-1}(E_k - E_j)^2|c_k|^2 |c_j|^2\nonumber\\
&\!\!\!\!\geq\!\!\!\!& \frac{1}{2}\sum_{k=s}^{D-1}\sum_{j=0}^{s-1}(E_k - E_j)^2|c_k|^2 |c_j|^2 \geq\frac{1}{2}\Delta_1^2 p(1-p),
\end{eqnarray}
where the last inequality follows from the fact that  $E_k - E_j \geq \Delta_1$ for all $k=s,\dots,D-1$ and $j=0,\dots, s-1$. Here, we define  $p := 1 - \sum_{k=0}^{s-1}|c_k|^2 = \sum_{k=s}^{D-1}|c_k|^2$. We further bound the term $p(1-p)$ in \eqref{expansion of the Riemannian gradient one specifc term}.
By the definition of $f(\mU^\star)$ and the normalization condition $f(\mU^\star) = E_0 = E_0\sum_{k=0}^{D-1}|c_k|^2$, we have
\begin{eqnarray}
\label{upper bound of p 1}
f(\mU) - f(\mU^\star) &\!\!\!\!=\!\!\!\!& \sum_{k=0}^{D-1}E_k|c_k|^2 - E_0\sum_{k=0}^{D-1}|c_k|^2\nonumber\\
&\!\!\!\!=\!\!\!\!& \sum_{k=s}^{D-1}(E_k - E_0 )|c_k|^2\nonumber\\
&\!\!\!\!\leq\!\!\!\!&(E_{D-1} - E_0 )\sum_{k=s}^{D-1}|c_k|^2\nonumber\\
&\!\!\!\!\leq\!\!\!\!& 2\|\mH\| p,
\end{eqnarray}
implying that $p\geq \frac{f(\mU) - f(\mU^\star)}{2\|\mH\|}$. In addition, since
\begin{eqnarray}
\label{lower bound of p 1}
f(\mU) - f(\mU^\star) &\!\!\!\!=\!\!\!\!& \sum_{k=s}^{D-1}(E_k - E_0 )|c_k|^2\nonumber\\
&\!\!\!\!\geq\!\!\!\!&\Delta_1\sum_{k=s}^{D-1}|c_k|^2\nonumber\\
&\!\!\!\!\geq\!\!\!\!& \Delta_1 p,
\end{eqnarray}
we can obtain $p\leq \frac{f(\mU) - f(\mU^\star)}{\Delta_1} $. Therefore, we have $\frac{f(\mU) - f(\mU^\star)}{2\|\mH\|} \leq  p\leq  \frac{f(\mU) - f(\mU^\star)}{\Delta_1}$. Assume that $f(\mU) - f(\mU^\star)\leq \frac{\Delta_1}{2}$. Then, it follows that $1 - p \geq \frac{1}{2}$. Consequently, we obtain the following lower bound for $F_{|\vphi\rangle}(\mH)$.
\begin{eqnarray}
\label{lower bound of F in the single layer}
\|\text{grad}_{\mU}f(\mU) \|_F^2  = \frac{1}{2}F_{|\vphi\rangle}(\mH) = \frac{1}{4}\Delta_1^2 p(1-p) \geq \frac{1}{8}\Delta_1^2 p\geq \frac{\Delta_1^2}{16\|\mH\|}(f(\mU) - f(\mU^\star) ).
\end{eqnarray}

Based on \eqref{L smooth condition for Riemmanin gradient}, it follows that
\begin{eqnarray}
\label{final conclusion of single circuit}
f(\mU^{t+1} ) - f(\mU^\star) &\!\!\!\!\leq\!\!\!\!& f(\mU^t ) - f(\mU^\star) - \frac{\mu}{2}\|\text{grad}_{\mU}f(\mU^t) \|_F^2\nonumber\\
&\!\!\!\!\leq\!\!\!\!&\bigg( 1 - \frac{\Delta_1^2\mu}{32\|\mH\|} \bigg)(f(\mU^t ) - f(\mU^\star))\nonumber\\
&\!\!\!\!\leq\!\!\!\!&\bigg( 1 - \frac{\Delta_1^2\mu}{32\|\mH\|} \bigg)^{t+1}(f(\mU^0 ) - f(\mU^\star)),
\end{eqnarray}
where we assume that $f(\mU^0) - f(\mU^\star)\leq \frac{\Delta_1}{2}$.

\end{proof}

\section{Proof of \Cref{Theorem: global optimization of one-layer quantum circuit}}
\label{Proof of global optimization}

We restate the objective function
\begin{eqnarray}
\label{objective function appendix}
f(\mU) = \trace(\mH \mU \vrho_0 \mU^\dagger) = \langle\vphi_0|\mU^\dagger\mH \mU |\vphi_0 \rangle.
\end{eqnarray}
as defined in \eqref{loss function of one layer}, and proceed to characterize its critical points. A unitary $\wh{\mU}_k$ is a critical point if and only if the Riemannian gradient vanishes:
\begin{eqnarray}
\label{first-order Riemannian gradient}
\text{grad}_{\mU}f(\mU) =\mH \mU \vrho_0 - \mU\text{sym}({\mU}^\dagger\mH\mU\vrho_0) = {\bm 0}.
\end{eqnarray}
which is equivalent to the commutation relation $\wh{\mU}_k^\dagger\mH\wh{\mU}_k\vrho_0 = \vrho_0 \wh{\mU}_k^\dagger\mH\wh{\mU}_k$. Since $\vrho_0 = |\vphi_0 \rangle \langle\vphi_0| $ is a rank-one projector, acting on $|\vphi_0 \rangle$ from the right yields $\wh{\mU}_k^\dagger \mH \wh{\mU}_k |\vphi_0 \rangle =\langle\vphi_0| \wh{\mU}_k^\dagger \mH \wh{\mU}_k |\vphi_0 \rangle|\vphi_0 \rangle$, showing that $|\vphi_0\rangle$ must be an eigenvector of $\wh{\mU}_k^\dagger\mH\wh{\mU}_k$.
Substituting the spectral decomposition $\mH = \sum_{j=0}^{D-1} E_j\,|\vpsi_j\rangle\langle\vpsi_j|$ gives
\begin{eqnarray}
    \sum_{j=0}^{D-1} E_j
    \left\langle\vpsi_j\big|\wh{\mU}_k|\vphi_0\rangle\,\wh{\mU}_k^\dagger|\vpsi_j\right\rangle
    \;=\;
    \bigg(\sum_{j=0}^{D-1} E_j\,
    \big|\langle\vpsi_j|\wh{\mU}_k|\vphi_0\rangle\big|^2\bigg)|\vphi_0\rangle,
\end{eqnarray}
which holds if and only if $\wh{\mU}_k|\vphi_0\rangle$ is itself an eigenvector of $\mH$.
Consequently, the critical points are precisely those unitaries satisfying
\begin{eqnarray}
\label{first-order condition}
    \wh{\mU}_k\,|\vphi_0\rangle = |\vpsi_k\rangle, k = 0, \dots, D-1.
\end{eqnarray}

Next, we characterize the Riemannian Hessian at each critical point.
According to \cite[Proposition~1]{lai2026quantum}, the Riemannian Hessian of $f$
at $\mU$ in the direction $\mOmega\mU \in \text{T}_{\mU}\calU(D)$ is given by
\begin{eqnarray}
\label{second-order Riemannian Hessian matrix}
    \text{Hess}f(\mU)[\mOmega\mU] :=  \frac{1}{2} ([\mH, [\mOmega, \mU\vrho_0\mU^\dagger ]  ] + [[\mH, \mOmega ], \mU\vrho_0\mU^\dagger ]  )\mU,
\end{eqnarray}
where $[\mA,\mB] := \mA\mB - \mB\mA$ denotes the matrix commutator and $\mOmega \in \mathfrak{u}(D)$ is a skew-Hermitian matrix. We now evaluate \eqref{second-order Riemannian Hessian matrix} at the critical point
$\wh{\mU}_k$ satisfying $\wh{\mU}_k|\vphi_0\rangle = |\vpsi_k\rangle$.
Define the conjugated state
\begin{equation}
    \wt{\vrho}_k
    \;:=\;
    \wh{\mU}_k\,\vrho_0\,\wh{\mU}_k^\dagger
    \;=\; |\vpsi_k\rangle\langle\vpsi_k|,
\end{equation}
where the second equality follows from the critical-point condition
\eqref{first-order condition}.
Expanding the matrix commutators in \eqref{second-order Riemannian Hessian matrix},
we obtain
\begin{align}
    \bigl[\mH,\bigl[\mOmega,\wt{\vrho}_k\bigr]\bigr]
    &= \mH\mOmega\wt{\vrho}_k - \mH\wt{\vrho}_k\mOmega
      - \mOmega\wt{\vrho}_k\mH + \wt{\vrho}_k\mOmega\mH, \\
    \bigl[\bigl[\mH,\mOmega\bigr],\wt{\vrho}_k\bigr]
    &= \mH\mOmega\wt{\vrho}_k - \mOmega\mH\wt{\vrho}_k
      - \wt{\vrho}_k\mH\mOmega + \wt{\vrho}_k\mOmega\mH.
\end{align}
Summing and dividing by $2$ yields
\begin{eqnarray}
\label{hessian expanded}
    \text{Hess}\,f(\wh{\mU}_k)[\mOmega\wh{\mU}_k]
    \;=\;
    \Bigl(
        \mH\mOmega\wt{\vrho}_k
        + \wt{\vrho}_k\mOmega\mH
        - \tfrac{1}{2}\mH\wt{\vrho}_k\mOmega
        - \tfrac{1}{2}\mOmega\mH\wt{\vrho}_k
        - \tfrac{1}{2}\wt{\vrho}_k\mH\mOmega
        - \tfrac{1}{2}\mOmega\wt{\vrho}_k\mH
    \Bigr)\wh{\mU}_k.
\end{eqnarray}
The bilinear form of the Riemannian Hessian can now be computed as follows:
\begin{eqnarray}
\label{hessian expanded bilinear}
    \text{Hess}\,f(\wh{\mU}_k)[\mOmega\wh{\mU}_k,\mOmega\wh{\mU}_k] &\!\!\!\!=\!\!\!\!& \bigl\langle \mOmega\wh{\mU}_k,\,
    \text{Hess}\,f(\wh{\mU}_k)[\mOmega\wh{\mU}_k]
    \bigr\rangle\nonumber\\
    &\!\!\!\!=\!\!\!\!& \trace\bigl((\mOmega\wh{\mU}_k)^\dagger\,
    \text{Hess}\,f(\wh{\mU}_k)[\mOmega\wh{\mU}_k]\bigr)\nonumber\\
    &\!\!\!\!=\!\!\!\!& \trace(\mOmega^\dagger\,\mX),
\end{eqnarray}
where, using the unitarity of $\wh{\mU}_k$ and the cyclic property of
the trace, we define
\begin{eqnarray}
    \mX := \mH\mOmega\wt{\vrho}_k
         + \wt{\vrho}_k\mOmega\mH
         - \tfrac{1}{2}\mH\wt{\vrho}_k\mOmega
         - \tfrac{1}{2}\mOmega\mH\wt{\vrho}_k
         - \tfrac{1}{2}\wt{\vrho}_k\mH\mOmega
         - \tfrac{1}{2}\mOmega\wt{\vrho}_k\mH.
\end{eqnarray}
Before computing \eqref{hessian expanded bilinear}, we expand all operators in the eigenbasis $\{|\vpsi_j\rangle\}$ of
$\mH$. Specifically, we use:
\begin{itemize}
    \item Spectral decomposition: $\mH = \sum_j E_j|\vpsi_j\rangle\langle\vpsi_j|$, so $\mH$ is diagonal with entries $H_{jl} = E_j\delta_{jl}$;
    \item Rank-one projector: $\wt{\vrho}_k = |\vpsi_k\rangle\langle\vpsi_k|$, with entries $(\wt{\vrho}_k)_{jl} = \delta_{jk}\delta_{lk}$;
    \item Matrix elements of $\mOmega$: $\Omega_{jl} := \langle\vpsi_j|\mOmega|\vpsi_l\rangle$, satisfying $\Omega_{jl} = -\Omega^{*}_{lj}$ (skew-Hermitian), so that $|\Omega_{jl}|^2 = |\Omega_{lj}|^2$ for all $j,l$.
\end{itemize}
Here $\delta_{jl}$ denotes the Kronecker delta, i.e., $\delta_{jl} = 1$ if $j=l$ and $\delta_{jl}=0$ otherwise. A key equation is used repeatedly. When $\mA$ and $\mB$ are
diagonal in the eigenbasis with entries $A_{jj}$ and $B_{ll}$,
\begin{eqnarray}
\label{eq:trace_diag}
    \trace(\mOmega^\dagger\mA\mOmega\mB)
    = \sum_{j,l} A_{jj}\,B_{ll}\,|\Omega_{jl}|^2.
\end{eqnarray}
Strictly speaking, this simplification is valid when at least one of the diagonal matrices is rank-one, as is the case for $\wt{\vrho}_k$. In that case, $(\wt{\vrho}_k)_{jl} = \delta_{jk}\delta_{lk}$ which eliminates all off-diagonal contributions and reduces each sum to the corresponding $k$-th index. We now evaluate each term in $\trace(\mOmega^\dagger\mX)$ in turn.

\begin{itemize}
\item{\textbf{Terms 1 and 2:} $\trace(\mOmega^\dagger\mH\mOmega\wt{\vrho}_k)$ and $\trace(\mOmega^\dagger\wt{\vrho}_k\mOmega\mH)$.} Applying \eqref{eq:trace_diag} with $\mA = \mH$ (diagonal, $A_{jj}=E_j$)
and $\mB = \wt{\vrho}_k$ (diagonal, $B_{ll}=\delta_{lk}$):
\begin{eqnarray}
\label{first term}
    \trace(\mOmega^\dagger\mH\mOmega\wt{\vrho}_k)
    = \sum_{j,l} E_j\,\delta_{lk}\,|\Omega_{jl}|^2
     = \sum_j E_j\,|\Omega_{jk}|^2.
\end{eqnarray}
Similarly, we have
\begin{eqnarray}
\label{second term}
    \trace(\mOmega^\dagger\wt{\vrho}_k\mOmega\mH) = \sum_{j,l} \delta_{jk} E_l |\Omega_{jl} |^2 = \sum_{l}  E_l |\Omega_{kl} |^2 = \sum_j E_j\,|\Omega_{jk}|^2.
\end{eqnarray}

\item{\textbf{Terms 3 and 5:} $-\tfrac{1}{2}\trace(\mOmega^\dagger\mH\wt{\vrho}_k\mOmega)$ and $-\tfrac{1}{2}\trace(\mOmega^\dagger\wt{\vrho}_k\mH\mOmega)$.} Since $\mH\wt{\vrho}_k = E_k\wt{\vrho}_k$ (as $|\vpsi_k\rangle$ is an eigenvector of $\mH$), we have
\begin{eqnarray}
\label{third term}
    -\tfrac{1}{2}\trace(\mOmega^\dagger\mH\wt{\vrho}_k\mOmega) &\!\!\!\! = \!\!\!\!& -\tfrac{E_k}{2}\trace(\mOmega^\dagger\wt{\vrho}_k\mOmega) = -\tfrac{E_k}{2}\sum_{j,l}\delta_{jk}\,|\Omega_{jl}|^2\nonumber\\
    &\!\!\!\! = \!\!\!\!& -\tfrac{E_k}{2}\sum_l|\Omega_{kl}|^2 = -\tfrac{E_k}{2}\sum_j|\Omega_{jk}|^2,
\end{eqnarray}
where the second equation uses \eqref{eq:trace_diag} with diagonal matrices $\wt{\vrho}_k$ and $\mId$. In addition, because of $\wt{\vrho}_k\mH = E_k\wt{\vrho}_k$, we also have
\begin{eqnarray}
\label{fifth term}
    -\tfrac{1}{2}\trace(\mOmega^\dagger\wt{\vrho}_k\mH\mOmega) = -\tfrac{E_k}{2}\trace(\mOmega^\dagger\wt{\vrho}_k\mOmega) = -\tfrac{E_k}{2}\sum_l|\Omega_{kl}|^2 = -\tfrac{E_k}{2}\sum_j|\Omega_{jk}|^2.
\end{eqnarray}

\item{\textbf{Terms 4 and 6:} $-\tfrac{1}{2}\trace(\mOmega^\dagger\mOmega\mH\wt{\vrho}_k)$ and $-\tfrac{1}{2}\trace(\mOmega^\dagger\mOmega\wt{\vrho}_k\mH)$.} Since  $\mH\wt{\vrho}_k = E_k\wt{\vrho}_k$ and $\wt{\vrho}_k\mH = E_k\wt{\vrho}_k$, we have
\begin{eqnarray}
\label{fourth term}
    -\tfrac{1}{2}\trace(\mOmega^\dagger\mOmega\mH\wt{\vrho}_k)
    = -\tfrac{E_k}{2}\trace(\mOmega^\dagger\mOmega\wt{\vrho}_k) = -\frac{E_k}{2}\sum_{j,l}\delta_{lk}\,|\Omega_{jl}|^2  = -\frac{E_k}{2}\sum_j|\Omega_{jk}|^2,
\end{eqnarray}
and
\begin{eqnarray}
\label{sixth term}
-\tfrac{1}{2}\trace(\mOmega^\dagger\mOmega\wt{\vrho}_k\mH)
    = -\tfrac{E_k}{2}\trace(\mOmega^\dagger\mOmega\wt{\vrho}_k)
    = -\tfrac{E_k}{2}\sum_j|\Omega_{jk}|^2.
\end{eqnarray}

\end{itemize}

By substituting all the previously derived equations into \eqref{hessian expanded bilinear}, we obtain
\begin{eqnarray}
\label{Bilinear form of Hessian final conclusion}
\text{Hess}\,f(\wh{\mU}_k)[\mOmega\wh{\mU}_k,\mOmega\wh{\mU}_k] &\!\!\!\! = \!\!\!\!& 2\sum_j E_j|\Omega_{jk}|^2
     - 2E_k\sum_j|\Omega_{jk}|^2
    \nonumber\\
    &\!\!\!\! = \!\!\!\!& 2\sum_{l=0}^{D-1}(E_l - E_k)\,|\Omega_{lk}|^2.
\end{eqnarray}
The $l=k$ term vanishes identically since $E_k - E_k = 0$, so \eqref{Bilinear form of Hessian final conclusion} depends only on the off-diagonal elements $\{\Omega_{lk}\}_{l\neq k}$. The sign of each summand $(E_l - E_k)|\Omega_{lk}|^2$ is determined solely by the relative ordering of the eigenvalues $E_l$ and $E_k$, which allows us to classify the nature of every critical point.
\begin{itemize}
    \item \textbf{Ground states} ($k = 0, \dots, s-1$): For these critical points, $E_k = E_0$ is the lowest eigenvalue, so $E_l - E_k \geq \Delta_1 > 0$ for all $l \geq s$, and $E_l - E_k = 0$ for all $l = 0,\ldots,s-1$ by the degeneracy assumption $E_0 = \cdots = E_{s-1}$. Hence every term in \eqref{Bilinear form of Hessian final conclusion} is non-negative, and $\text{Hess}\,f(\wh{\mU}_k)$ is positive semidefinite. The zero directions correspond to perturbations within the degenerate ground-state subspace, which do not change the energy to second order. Since $f(\wh{\mU}_k) = E_k = E_0 \leq f(\mU)$ for all $\mU\in\mathcal{U}(D)$ by the variational principle, every $\wh{\mU}_k$ with $k = 0,\ldots,s-1$ is a \emph{global minimum}.

    \item \textbf{Excited states} ($k \geq s$): For these critical points, there exist indices $l < k$ satisfying $E_l < E_k$, so $E_l - E_k < 0$. Choosing $\mOmega$ such that $\Omega_{lk} \neq 0$ for any such $l$ while all other off-diagonal entries vanish, \eqref{Bilinear form of Hessian final conclusion} yields
    \begin{equation}
        \text{Hess}\,f(\wh{\mU}_k)[\mOmega\wh{\mU}_k,\mOmega\wh{\mU}_k]
        = 2(E_l - E_k)|\Omega_{lk}|^2 < 0.
    \end{equation}
    Hence the Hessian has at least one strictly negative eigenvalue, and $\wh{\mU}_k$ is a \emph{strict saddle point}.
\end{itemize}

\section{Proof of \Cref{Theorem: convergence analysis of multi circuit}}
\label{Proof of convergence analysis of multi circuit appendix}

Before analyzing the convergence rate of Riemannian gradient descent, we establish that the Riemannian gradient $\text{grad}_{\mU_h}f$ is $C$-Lipschitz continuous. Specifically, for any $(\mU_1,\ldots,\mU_N)$ and $(\mV_1,\ldots,\mV_N)$ with $\mU_h, \mV_h \in \calU(D)$ for all $h \in [N]$, we have
\begin{eqnarray}
\label{Lipschitz continuous of multi layer quantum circuit}
&\!\!\!\!\!\!\!\!& \|\text{grad}_{\mU_h}g(\mV_1,\ldots,\mV_N) - \text{grad}_{\mU_h}g(\mU_1,\ldots,\mU_N)\|_F\nonumber\\
&\!\!\!\! = \!\!\!\!&  \|\nabla_{\mU_h}g(\mU_1,\ldots,\mU_N) - \nabla_{\mU_h}g(\mV_1,\ldots,\mV_N)\nonumber\\
&\!\!\!\!\!\!\!\!& - \mU_h\text{sym}({\mU_h}^\dagger\nabla_{\mU_h}g(\mU_1,\ldots,\mU_N)) + \mV_h\text{sym}({\mV_h}^\dagger\nabla_{\mU_h}g(\mV_1,\ldots,\mV_N))\|_F\nonumber\\
&\!\!\!\!\leq\!\!\!\!&\| \nabla_{\mU_h}g(\mU_1,\ldots,\mU_N) - \nabla_{\mU_h}g(\mV_1,\ldots,\mV_N)  \|_F \nonumber\\
&\!\!\!\!\!\!\!\!& + \|\mU_h\text{sym}({\mU_h}^\dagger\nabla_{\mU_h}g(\mU_1,\ldots,\mU_N)) - \mV_h\text{sym}({\mV_h}^\dagger\nabla_{\mU_h}g(\mV_1,\ldots,\mV_N))\|_F.
\end{eqnarray}

First, we can derive
{\small \begin{eqnarray}
\label{Lipschitz continuous of multi layer quantum circuit first term}
&\!\!\!\!\!\!\!\!& \| \nabla_{\mU_h}g(\mU_1,\ldots,\mU_N) - \nabla_{\mU_h}g(\mV_1,\ldots,\mV_N)  \|_F\nonumber\\
&\!\!\!\! = \!\!\!\!& \| \mU_{h-1}^\dagger\cdots\mU_1^\dagger\mH \mU_1\cdots \mU_N\vrho_0 \mU_{N}^\dagger\cdots\mU_{h+1}^\dagger - \mV_{h-1}^\dagger\cdots\mV_1^\dagger\mH \mV_1\cdots \mV_N\vrho_0 \mV_{N}^\dagger\cdots\mV_{h+1}^\dagger  \|_F\nonumber\\
&\!\!\!\! = \!\!\!\!& \bigg\|\sum_{a=1}^{h-1} \mV_{h-1}^\dagger\cdots \mV_{a+1}^\dagger(\mU_a - \mV_a )^\dagger \mU_{a-1}^\dagger\cdots \mU_1^\dagger\mH \mU_1\cdots \mU_N\vrho_0 \mU_{N}^\dagger\cdots\mU_{h+1}^\dagger\nonumber\\
&\!\!\!\!\!\!\!\!& + \mV_{h-1}^\dagger\cdots \mV_1^\dagger\mH \mU_1\cdots \mU_N\vrho_0 \mU_{N}^\dagger\cdots\mU_{h+1}^\dagger -  \mV_{h-1}^\dagger\cdots\mV_1^\dagger\mH \mV_1\cdots \mV_N\vrho_0 \mV_{N}^\dagger\cdots\mV_{h+1}^\dagger   \bigg\|_F\nonumber\\
&\!\!\!\! = \!\!\!\!& \bigg\|\sum_{a=1}^{h-1} \mV_{h-1}^\dagger\cdots \mV_{a+1}^\dagger(\mU_a - \mV_a )^\dagger \mU_{a-1}^\dagger\cdots \mU_1^\dagger\mH \mU_1\cdots \mU_N\vrho_0 \mU_{N}^\dagger\cdots\mU_{h+1}^\dagger\nonumber\\
&\!\!\!\!\!\!\!\!& +\sum_{b=1}^{N} \mV_{h-1}^\dagger\cdots \mV_1^\dagger\mH \mV_1\cdots \mV_{b-1}(\mU_b - \mV_b)\mU_{b+1}\cdots \mU_N  \vrho_0 \mU_{N}^\dagger\cdots\mU_{h+1}^\dagger\nonumber\\
&\!\!\!\!\!\!\!\!& + \mV_{h-1}^\dagger\cdots \mV_1^\dagger\mH \mV_1\cdots \mV_N\vrho_0 \mU_{N}^\dagger\cdots\mU_{h+1}^\dagger -  \mV_{h-1}^\dagger\cdots\mV_1^\dagger\mH \mV_1\cdots \mV_N\vrho_0 \mV_{N}^\dagger\cdots\mV_{h+1}^\dagger \bigg\|_F\nonumber\\
&\!\!\!\! = \!\!\!\!& \bigg\|\sum_{a=1}^{h-1} \mV_{h-1}^\dagger\cdots \mV_{a+1}^\dagger(\mU_a - \mV_a )^\dagger \mU_{a-1}^\dagger\cdots \mU_1^\dagger\mH \mU_1\cdots \mU_N\vrho_0 \mU_{N}^\dagger\cdots\mU_{h+1}^\dagger\nonumber\\
&\!\!\!\!\!\!\!\!& +\sum_{b=1}^{N} \mV_{h-1}^\dagger\cdots \mV_1^\dagger\mH \mV_1\cdots \mV_{b-1}(\mU_b - \mV_b)\mU_{b+1}\cdots \mU_N  \vrho_0 \mU_{N}^\dagger\cdots\mU_{h+1}^\dagger\nonumber\\
&\!\!\!\!\!\!\!\!& + \sum_{c=h+1}^{N} \mV_{h-1}^\dagger\cdots\mV_1^\dagger\mH \mV_1\cdots \mV_N\vrho_0 \mV_{N}^\dagger\cdots \mV_{c+1}^\dagger(\mU_c - \mV_c)^\dagger\mU_{c-1}\cdots\mU_{h+1}^\dagger\bigg\|_F\nonumber\\
&\!\!\!\! \leq \!\!\!\!& \sum_{a=1}^{h-1} \|\mU_a - \mV_a\|_F \|\mH\| + \sum_{b=1}^{N} \|\mU_b - \mV_b\|_F \|\mH\| + \sum_{c=h+1}^{N} \|\mU_c - \mV_c\|_F \|\mH\|\nonumber\\
&\!\!\!\! = \!\!\!\!&2\|\mH\|\sum_{b=1}^{N} \|\mU_b - \mV_b\|_F - \|\mH\| \|\mU_h - \mV_h\|_F.
\end{eqnarray}}

In addition, we have
\begin{eqnarray}
\label{Lipschitz continuous of multi layer quantum circuit second term}
&\!\!\!\!\!\!\!\!&\|\mU_h\text{sym}({\mU_h}^\dagger\nabla_{\mU_h}g(\mU_1,\ldots,\mU_N)) - \mV_h\text{sym}({\mV_h}^\dagger\nabla_{\mU_h}g(\mV_1,\ldots,\mV_N))\|_F\nonumber\\
&\!\!\!\!\leq\!\!\!\!& \|\mU_h\text{sym}({\mU_h}^\dagger\nabla_{\mU_h}g(\mU_1,\ldots,\mU_N)) - \mV_h\text{sym}({\mU_h}^\dagger\nabla_{\mU_h}g(\mU_1,\ldots,\mU_N))\|_F\nonumber\\
&\!\!\!\!\!\!\!\!& + \|\mV_h\text{sym}({\mU_h}^\dagger\nabla_{\mU_h}g(\mU_1,\ldots,\mU_N)) - \mV_h\text{sym}({\mV_h}^\dagger\nabla_{\mU_h}g(\mV_1,\ldots,\mV_N))\|_F\nonumber\\
&\!\!\!\!\leq\!\!\!\!& \|\mU_h - \mV_h \|_F \|\text{sym}({\mU_h}^\dagger\nabla_{\mU_h}g(\mU_1,\ldots,\mU_N))\|\nonumber\\
&\!\!\!\!\!\!\!\!& + \|\text{sym}({\mU_h}^\dagger\nabla_{\mU_h}g(\mU_1,\ldots,\mU_N)) - \text{sym}({\mV_h}^\dagger\nabla_{\mU_h}g(\mV_1,\ldots,\mV_N))\|_F\nonumber\\
&\!\!\!\!\leq\!\!\!\!& \frac{1}{2}\|\mU_h - \mV_h \|_F\|\mU_h^\dagger\cdots\mU_1^\dagger\mH \mU_1\cdots \mU_N\vrho_0 \mU_{N}^\dagger\cdots\mU_{h+1}^\dagger + \mU_{h+1}\cdots\mU_{N}\vrho_0 \mU_N^\dagger \cdots\mU_1^\dagger \mH\mU_1 \cdots \mU_h \|\nonumber\\
&\!\!\!\!\!\!\!\!& + \|{\mU_h}^\dagger\nabla_{\mU_h}g(\mU_1,\ldots,\mU_N) -  {\mV_h}^\dagger\nabla_{\mU_h}g(\mV_1,\ldots,\mV_N) \|_F\nonumber\\
&\!\!\!\!\leq\!\!\!\!& \|\mH\| \|\mU_h - \mV_h \|_F + \| (\mU_h - \mV_h)^\dagger\nabla_{\mU_h}g(\mU_1,\ldots,\mU_N)  \|_F\nonumber\\
&\!\!\!\!\!\!\!\!& + \|\mV_h^\dagger(\nabla_{\mU_h}g(\mU_1,\ldots,\mU_N) - \nabla_{\mU_h}g(\mV_1,\ldots,\mV_N) ) \|_F\nonumber\\
&\!\!\!\!\leq\!\!\!\!& 2\|\mH\|\sum_{b=1}^{N} \|\mU_b - \mV_b\|_F + \|\mH\| \|\mU_h - \mV_h\|_F,
\end{eqnarray}
where the last line uses $\|\nabla_{\mU_h}g(\mU_1,\ldots,\mU_N)\|\leq \|\mH\|$ and \eqref{Lipschitz continuous of multi layer quantum circuit first term}.

By combining \eqref{Lipschitz continuous of multi layer quantum circuit first term} and \eqref{Lipschitz continuous of multi layer quantum circuit second term}, \eqref{Lipschitz continuous of multi layer quantum circuit} can be expressed as follows:
\begin{eqnarray}
\label{Lipschitz continuous of multi layer quantum circuit final conclusion}
 \|\text{grad}_{\mU_h}g(\mV_1,\ldots,\mV_N) - \text{grad}_{\mU_h}g(\mU_1,\ldots,\mU_N)\|_F &\!\!\!\! \leq \!\!\!\!&  4\|\mH\|\sum_{b=1}^{N} \|\mU_b - \mV_b\|_F\nonumber\\
&\!\!\!\! \leq \!\!\!\!&  4\sqrt{N}\|\mH\|\sqrt{\sum_{b=1}^{N} \|\mU_b - \mV_b\|_F^2}.
\end{eqnarray}
Hence, by applying \Cref{lemma:Riemannian descent conclusion} with $L = \frac{8N+1}{2}\|\mH\|$, we  obtain
\begin{eqnarray}
\label{descent direction of multi layer quantum circuit final conclusion}
 g(\mU_1^{t+1},\ldots,\mU_N^{t+1})
    &\!\!\!\!\leq\!\!\!\!& g(\mU_1^t,\ldots,\mU_N^t)
    - \mu(1 - L\mu)\sum_{h=1}^{N}\| \text{grad}_{\mU_h}g(\mU_1^t,\ldots,\mU_N^t) \|_F^2\nonumber\\
    &\!\!\!\!\leq\!\!\!\!&g(\mU_1^t,\ldots,\mU_N^t)
    - \frac{\mu}{2}\sum_{h=1}^{N}\| \text{grad}_{\mU_h}g(\mU_1^t,\ldots,\mU_N^t) \|_F^2,
\end{eqnarray}
where $\mu\leq \frac{1}{2L}$ is chosen. In the following, we derive a lower bound for the squared Frobenius norm of the Riemannian gradient $ \text{grad}_{\mU_h}g(\mU_1,\ldots,\mU_N) $. Let $\mA_h = \mU_1\cdots \mU_{h-1}$, $\mB_h = \mU_{h+1}\cdots \mU_N$, $\wt\mH_h = \mA_h^\dagger \mH \mA_h$ and $\wt\vrho_h = \mB_h\vrho_0\mB_h^\dagger$. With these definitions, the gradient $\nabla_{\mU_h}g(\mU_1,\ldots,\mU_N)$ can be expressed as
\begin{eqnarray}
\label{Another expression of gradient of multi layers}
\nabla_{\mU_h}g(\mU_1,\ldots,\mU_N) = \mA_h^\dagger \mH \mA_h \mU_h \mB_h \vrho_0 \mB_h^\dagger = \wt\mH_h \mU_h \wt\vrho_h.
\end{eqnarray}
We define $|\vphi_h \> = \mU_h\mB_h |\vphi_0 \>$. Using the approach outlined in \eqref{expansion of the Riemannian gradient}, it follows that
\begin{eqnarray}
\label{expansion of the Riemannian gradient multilayer}
\|\nabla_{\mU_h}g(\mU_1,\ldots,\mU_N) \|_F^2 =  \frac{\langle\vphi_h| \wt\mH_h^2  |\vphi_h \rangle - (g(\mU_1,\ldots,\mU_N))^2}{2} = \frac{1}{2}F_{|\vphi_h \rangle}(\wt\mH).
\end{eqnarray}
Therefore, we can further obtain
\begin{eqnarray}
\label{lower bound of the summation of Riemannian gradient multilayer}
\sum_{h=1}^{N}\| \text{grad}_{\mU_h}g(\mU_1,\ldots,\mU_N) \|_F^2 &\!\!\!\!\geq\!\!\!\!& \|\nabla_{\mU_1}g(\mU_1,\ldots,\mU_N) \|_F^2\nonumber\\
&\!\!\!\!=\!\!\!\!& \frac{1}{2}F_{|\vphi_1 \rangle}(\wt\mH)\nonumber\\
&\!\!\!\!=\!\!\!\!&  \frac{\langle\vphi_h|\mU_N^\dagger\cdots \mU_1^\dagger \mH^2 \mU_1\cdots \mU_N  |\vphi_h \rangle - (g(\mU_1,\ldots,\mU_N))^2}{2}\nonumber\\
&\!\!\!\!\geq\!\!\!\!&\frac{\Delta_1^2}{16\|\mH\|}(g(\mU_1,\dots,\mU_N) - g(\mU_1^\star,\dots,\mU_N^\star) ),
\end{eqnarray}
where the last line follows the same derivation of \eqref{lower bound of F in the single layer} when $g(\mU_1,\dots,\mU_N) - g(\mU_1^\star,\dots,\mU_N^\star)\leq \frac{\Delta_1}{2}$.

Finally, using \eqref{descent direction of multi layer quantum circuit final conclusion}, we have
\begin{eqnarray}
\label{descent direction of multi layer quantum circuit final conclusion1}
 g(\mU_1^{t+1},\ldots,\mU_N^{t+1})- g(\mU_1^\star,\dots,\mU_N^\star)
    &\!\!\!\!\leq\!\!\!\!& \bigg(1 -\frac{\Delta_1^2\mu}{32\|\mH\|}\bigg)(g(\mU_1^t,\dots,\mU_N^t) - g(\mU_1^\star,\dots,\mU_N^\star) )\nonumber\\
    &\!\!\!\!\leq\!\!\!\!&\bigg(1 -\frac{\Delta_1^2\mu}{32\|\mH\|}\bigg)^{t+1}(g(\mU_1^0,\dots,\mU_N^0) - g(\mU_1^\star,\dots,\mU_N^\star) ),
\end{eqnarray}
where we assume that $g(\mU_1^0,\dots,\mU_N^0) - g(\mU_1^\star,\dots,\mU_N^\star)\leq \frac{\Delta_1}{2}$.

\section{Proof of Eq.~\eqref{necessary condition of required layers main paper}}
\label{proof of dimension of multilayers}

We consider a layered variational ansatz of the form
\begin{eqnarray}
\mU(\vtheta) = \mU_1(\theta_1)\cdots \mU_N(\theta_N),
\end{eqnarray}
where each layer $\mU_h(\theta_h)$ is a smooth parametrization of an ansatz manifold $\mathcal A\subseteq \mathrm{SU}(D)$ with parameter $\theta_h\in\mathbb R^p$, where $p=\dim(\mathcal A)$ denotes the intrinsic dimension of the ansatz manifold.

\paragraph{Step 1: Product parameter space.}
Let
\begin{eqnarray}
\vtheta := (\theta_1,\ldots,\theta_N),
\end{eqnarray}
where each layer parameter $\theta_h\in\setR^p$. The full parameter space is therefore the Cartesian product
\begin{eqnarray}
\Theta := \underbrace{\setR^p\times\cdots\times\setR^p}_{N\text{ times}}
\cong \setR^{Np}.
\end{eqnarray}
By the product manifold theorem \cite[Example~1.8]{lee2003smooth}, $\Theta$ inherits a smooth manifold structure as the product of $N$ Euclidean manifolds, and its dimension is
\begin{eqnarray}
\dim(\Theta)=Np.
\end{eqnarray}

\paragraph{Step 2: Definition of the variational map.}
We define the variational map
\begin{eqnarray}
\mPhi:\Theta\to \mathrm{SU}(D)
\end{eqnarray}
by
\begin{eqnarray}
\mPhi(\theta_1,\ldots,\theta_N) := \mU_1(\theta_1)\cdots \mU_N(\theta_N).
\end{eqnarray}
Since each $\mU_h:\setR^p\to \mathrm{SU}(D)$ is smooth by assumption and the group multiplication in $\mathrm{SU}(D)$ is smooth (as $\mathrm{SU}(D)$ is a Lie group), the composite map $\mPhi$ is smooth.

\paragraph{Step 3: Differential and rank bound.}
To characterize how infinitesimal perturbations in the parameter space propagate to the unitary manifold, we study the differential of the variational map $\mPhi$. Since $\mPhi$ is smooth, its differential (also called the pushforward) at any point $\vtheta\in\Theta$ is a linear map between tangent spaces\footnote{The \emph{tangent space} $T_x M$ of a smooth manifold $M$ at a point $x$ is the vector space of all ``infinitesimal directions of motion'' at $x$. For Euclidean space $\Theta$, the tangent space at every point is canonically isomorphic to $\Theta$ itself. For the Lie group $\mathrm{SU}(D)$, the tangent space at the identity is the Lie algebra $\mathfrak{su}(D)$ (the space of traceless
anti-Hermitian $D\times D$ matrices), which has dimension $D^2-1$.}:
\begin{eqnarray}
\mathrm{d}\mPhi_{\vtheta} : T_{\vtheta}\Theta \cong \Theta  \to T_{\mPhi(\vtheta)}\mathrm{SU}(D),
\end{eqnarray}
where $\cong$ denotes a canonical vector space isomorphism\footnote{The symbol $\cong$ denotes an \emph{isomorphism}, i.e., a structure-preserving bijection. Here $T_{\vtheta}\Theta\cong\Theta$ means that the abstract tangent space and the concrete Euclidean space are isomorphic as vector spaces. They are technically distinct mathematical objects, but have identical algebraic structure.}. Concretely, the differential maps a tangent vector $\vv\in T_{\vtheta}\Theta\cong\Theta$ to the
tangent vector in $T_{\mPhi(\vtheta)}\mathrm{SU}(D)$ given by\footnote{This is the manifold generalization of the directional derivative from multivariable calculus. The vector $\mathrm{d}\mPhi_{\vtheta}(\vv)$ describes how the unitary $\mPhi(\vtheta)$ changes when the parameter $\vtheta$ is perturbed infinitesimally in the direction $\vv$.}
\begin{eqnarray}
\mathrm{d}\mPhi_{\vtheta}(\vv) = \lim_{t\to 0}\frac{\mPhi(\vtheta+t\vv)-\mPhi(\vtheta)}{t}.
\end{eqnarray}

The \emph{rank} of the differential is defined as the dimension of its image\footnote{For a linear map $A:V\to W$, the rank is $\mathrm{rank}(A):=\dim(\mathrm{Im}(A))$. It counts the number of linearly independent output directions. If $\mathrm{rank}(\mathrm{d}\mPhi_{\vtheta}) = Np$, all parameter directions produce independent changes in the unitary; if $\mathrm{rank}(\mathrm{d}\mPhi_{\vtheta}) < Np$, some parameter directions are redundant.}:
\begin{eqnarray}
\mathrm{rank}(\mathrm{d}\mPhi_{\vtheta}) := \dim\!\big(\mathrm{Im}(\mathrm{d}\mPhi_{\vtheta})\big).
\end{eqnarray}
By a standard result in linear algebra, the rank of any linear map $A:V\to W$ satisfies $\mathrm{rank}(A)\leq\dim(V)$. Applying this to $A=\mathrm{d}\mPhi_{\vtheta}$ with $V = \Theta$, we obtain
\begin{eqnarray}
\mathrm{rank}(\mathrm{d}\mPhi_{\vtheta})\leq \dim(\Theta) = Np \quad \text{for all } \vtheta\in\Theta.
\end{eqnarray}

\paragraph{Step 4: Reachable set and dimension bound.}
We define the \emph{reachable set} of the variational ansatz as the image of $\mPhi$:
\begin{eqnarray}
\mathcal{M}_N := \mPhi(\Theta) \subseteq \mathrm{SU}(D).
\end{eqnarray}

For any $\vtheta\in\Theta$, the differential $\mathrm{d}\mPhi_{\vtheta}$ is a linear map between finite-dimensional vector spaces, so its rank is bounded by the dimensions of both the domain and the codomain:
\begin{eqnarray}\label{eq:rank-min-bound}
\mathrm{rank}(\mathrm{d}\mPhi_{\vtheta})
\leq \min\!\big(\dim(\Theta),\;
\dim(\mathrm{SU}(D))\big)
= \min(Np,\;D^2-1).
\end{eqnarray}

At any point $\vtheta_0$ where the rank function $\vtheta\mapsto\mathrm{rank}(\mathrm{d}\mPhi_{\vtheta})$ is locally constant, the constant rank theorem \cite[Theorem~4.12]{lee2003smooth} implies that the image $\mPhi(\mV)$ of a sufficiently small neighborhood $\mV$ of $\vtheta_0$ is a smooth embedded submanifold of $\mathrm{SU}(D)$, and its dimension equals the rank:
\begin{eqnarray}\label{eq:local-dim-rank}
\dim_{\mathrm{loc}}(\mathcal{M}_N,\,\mPhi(\vtheta_0))
= \mathrm{rank}(\mathrm{d}\mPhi_{\vtheta_0}).
\end{eqnarray}
In particular, this holds at every regular point $\vtheta_0$ (where $\mathrm{rank}(\mathrm{d}\mPhi_{\vtheta_0}) = \max_{\vtheta}\mathrm{rank}(\mathrm{d}\mPhi_{\vtheta})$), since the rank is integer-valued and upper semicontinuous, and therefore locally constant at its maximum.

We define the dimension of the reachable set as the supremum of the local dimension\footnote{This definition avoids assuming \emph{a priori} that $\mathcal{M}_N$ carries a global smooth manifold structure.}:
\begin{eqnarray}\label{eq:dim-def}
\dim(\mathcal{M}_N)
:= \sup_{\mU\in\mathcal{M}_N}\,
\dim_{\mathrm{loc}}(\mathcal{M}_N,\,\mU).
\end{eqnarray}
Since every $\mU\in\mathcal{M}_N$ is of the form $\mU=\mPhi(\vtheta)$ for some $\vtheta$, and \eqref{eq:local-dim-rank} gives $\dim_{\mathrm{loc}}(\mathcal{M}_N,\,\mPhi(\vtheta))
=\mathrm{rank}(\mathrm{d}\mPhi_{\vtheta})$ wherever the rank is locally constant, the supremum in \eqref{eq:dim-def} is achieved at regular points and equals the maximal rank:
\begin{eqnarray}\label{eq:dim-equals-max-rank}
\dim(\mathcal{M}_N)
= \sup_{\vtheta\in\Theta}\,
\mathrm{rank}(\mathrm{d}\mPhi_\theta).
\end{eqnarray}

Combining \eqref{eq:dim-equals-max-rank} with the rank bound
\eqref{eq:rank-min-bound} yields
\begin{eqnarray}\label{eq:refined-bound}
\dim(\mathcal{M}_N) \leq \min(Np,\;D^2-1).
\end{eqnarray}

\paragraph{Step 5: Final bound.}
Combining the above, the dimension of the reachable set satisfies
\begin{eqnarray}
\dim(\mathcal{M}_N)\leq Np.
\end{eqnarray}
In particular, since $\dim(\mathrm{SU}(D)) = D^2 - 1$, a necessary condition for the ansatz to be \emph{universal} (i.e., $\mathcal{M}_N = \mathrm{SU}(D)$) is
\begin{eqnarray}
Np \geq D^2 - 1,
\qquad\text{i.e.,}\qquad
N \geq \left\lceil \frac{D^2 - 1}{p}\right\rceil.
\end{eqnarray}

\paragraph{Step 6: Structured target regimes.}
We note, however, that in many practical settings the target unitary does not require exploration of the entire unitary group, but rather lies within a structured subset
\begin{eqnarray}
\mathcal{U}_{\mathrm{target}} \subseteq \mathrm{SU}(D).
\end{eqnarray}
For example, this may correspond to symmetry-preserving unitaries, commuting unitary families, or other physically constrained transformations. Let
\begin{eqnarray}
d_{\mathrm{target}}
:=
\dim(\mathcal{U}_{\mathrm{target}})
\end{eqnarray}
denote the intrinsic dimension of the target unitary family. Since the reachable set of the variational ansatz satisfies
\begin{eqnarray}
\mathcal{M}_N = \mPhi(\Theta) \subseteq \mathcal{U}_{\mathrm{target}},
\end{eqnarray}
its dimension is bounded by both the parameter dimension and the target manifold dimension:
\begin{eqnarray}
\dim(\mathcal{M}_N)
\leq
\min(Np,\;d_{\mathrm{target}}).
\end{eqnarray}
Consequently, in order for the ansatz to be expressive enough to cover the target unitary family, it is necessary that
\begin{eqnarray}
Np \geq d_{\mathrm{target}},
\end{eqnarray}
which implies the depth lower bound
\begin{eqnarray}
N \geq \left\lceil \frac{d_{\mathrm{target}}}{p}\right\rceil.
\end{eqnarray}

\section{Proof of \Cref{Theorem: initialization design of multi circuit}}
\label{Proof of initialization requirement appendix}

We begin by computing the expectation in the single-unitary case under the proposed initialization $\mU^0 = e^{-i\theta\mP}$. Specifically, we have
\begin{eqnarray}
\label{loss function of one layer initlaizatoin}
f(\mU^0) &\!\!\!\!=\!\!\!\!&  \langle\vphi_0|{\mU^0}^\dagger\mH \mU^0 |\vphi_0 \rangle\nonumber\\
&\!\!\!\!=\!\!\!\!& \langle\vphi_0|(\cos(\theta)\mId + i\sin(\theta) \mP)\mH (\cos(\theta)\mId - i\sin(\theta) \mP) |\vphi_0 \rangle\nonumber\\
&\!\!\!\!=\!\!\!\!& \cos^2(\theta)\langle\vphi_0| \mH |\vphi_0 \rangle + \sin^2(\theta)\langle\vphi_0|\mP \mH \mP |\vphi_0 \rangle + i \sin(\theta)\cos(\theta)\langle\vphi_0|(\mP \mH - \mH\mP) |\vphi_0 \rangle.
\end{eqnarray}
Since $\E[\cos^2(\theta) ] =\frac{1+e^{-2\sigma^2}}{2}$, $\E[\sin^2(\theta) ] =\frac{1 - e^{-2\sigma^2}}{2}$ and $\E[\sin(\theta)\cos(\theta)] = 0$, we can derive
\begin{eqnarray}
\label{loss function of one layer initlaizatoin expectation}
\E[f(\mU^0)] =  \frac{1+e^{-2\sigma^2}}{2}\langle\vphi_0| \mH |\vphi_0 \rangle + \frac{1-e^{-2\sigma^2}}{2}\langle\vphi_0|\mP \mH \mP |\vphi_0 \rangle.
\end{eqnarray}

Next, we consider the product-unitary case. Define $|\vphi_h \> = \mU_h^0\cdots \mU_N^0|\vphi_0 \>$ for $h \in [N]$, and let $\alpha = \frac{1+e^{-2\sigma^2}}{2}$ and $\beta = \frac{1-e^{-2\sigma^2}}{2}$. Since $\theta_1$ is independent of $\{\theta_2,\dots,\theta_N\}$, it holds that
\begin{eqnarray}
\label{loss function of one layer initlaizatoin expectation theta1}
\E_{\theta_1}[g(\mU_1^0,\dots, \mU_N^0)] &\!\!\!\!=\!\!\!\!&  \E_{\theta_1}[ \< \vphi_2| {\mU_1^0}^\dagger \mH \mU_1^0| \vphi_2  \> ]\nonumber\\
&\!\!\!\!=\!\!\!\!&  \frac{1+e^{-2\sigma^2}}{2} \langle\vphi_2| \mH |\vphi_2 \rangle + \frac{1-e^{-2\sigma^2}}{2} \langle\vphi_2|\mP_1 \mH \mP_1 |\vphi_2 \rangle\nonumber\\
&\!\!\!\!=\!\!\!\!& \alpha \langle\vphi_2| \mH |\vphi_2 \rangle + \beta \langle\vphi_2|\mP_1 \mH \mP_1 |\vphi_2 \rangle\nonumber\\
&\!\!\!\!=\!\!\!\!& \langle\vphi_2| \mH_1 |\vphi_2 \rangle,
\end{eqnarray}
where the second equality follows from \eqref{loss function of one layer initlaizatoin expectation} and we define $\mH_1 = \alpha\mH + \beta\mP_1 \mH \mP_1$. By applying the same analysis and taking expectations with respect to $\theta_2,\dots,\theta_N$ sequentially, we obtain
\begin{eqnarray}
\label{loss function of one layer initlaizatoin expectation thetaN}
\E_{\theta_1,\dots,\theta_N}[g(\mU_1^0,\dots, \mU_N^0)]  = \langle\vphi_0| \mH_N |\vphi_0 \rangle,
\end{eqnarray}
where $\mH_k = \alpha\mH_{k-1} + \beta\mP_k\mH_{k-1}\mP_k, k=1,\dots, N$ with $\mH_0 = \mH$. Hence, we have
\begin{eqnarray}
\label{loss function of one layer initlaizatoin expectation thetaN final}
\E_{\theta_1,\dots,\theta_N}[g(\mU_1^0,\dots, \mU_N^0)] - g(\mU_1^\star,\dots, \mU_N^\star)  = \langle\vphi_0| \mH_N |\vphi_0 \rangle -  g(\mU_1^\star,\dots, \mU_N^\star).
\end{eqnarray}
Let $e_k = \langle\vphi_0| \mH_k |\vphi_0 \rangle - g(\mU_1^\star,\dots, \mU_N^\star)$. Then we have
\begin{eqnarray}
\label{expansion of ek}
e_k &\!\!\!\!=\!\!\!\!& \alpha e_{k-1} + \beta\big(\langle\vphi_0| \mP_k\mH_{k-1}\mP_k|\vphi_0 \rangle -  g(\mU_1^\star,\dots, \mU_N^\star)\big)\nonumber\\
&\!\!\!\!\leq\!\!\!\!& \alpha e_{k-1} + \beta\|\mP_k\mH_{k-1}\mP_k - g(\mU_1^\star,\dots, \mU_N^\star)\mId  \|\nonumber\\
&\!\!\!\!-\!\!\!\!& \alpha e_{k-1} + \beta\|\mH_{k-1} - g(\mU_1^\star,\dots, \mU_N^\star)\mId  \|.
\end{eqnarray}
We prove by induction that
\begin{eqnarray}
\label{expansion of ek bound}
\|\mH_{k} - g(\mU_1^\star,\dots, \mU_N^\star)\mId  \|\leq \|\mH - g(\mU_1^\star,\dots, \mU_N^\star)\mId\|.
\end{eqnarray}
For $k=0$, the inequality holds trivially. Assume that it holds for $k=\ell-1$. Then, for $k=\ell$, we have
\begin{eqnarray}
\label{expansion of ek bound induction}
\|\mH_{\ell} - g(\mU_1^\star,\dots, \mU_N^\star)\mId  \| &\!\!\!\!\leq\!\!\!\!& \alpha\|\mH_{\ell-1} - g(\mU_1^\star,\dots, \mU_N^\star)\mId \| + \beta\|\mP_k( \mH_{\ell-1} - g(\mU_1^\star,\dots, \mU_N^\star)\mId    )\mP_k   \|\nonumber\\
&\!\!\!\!\leq\!\!\!\!&\alpha\|\mH_{\ell-1} - g(\mU_1^\star,\dots, \mU_N^\star)\mId \| + \beta\| \mH_{\ell-1} - g(\mU_1^\star,\dots, \mU_N^\star)\mId  \|\nonumber\\
&\!\!\!\! = \!\!\!\!&\|\mH_{\ell-1} - g(\mU_1^\star,\dots, \mU_N^\star)\mId \|\nonumber\\
&\!\!\!\!\leq\!\!\!\!& \|\mH - g(\mU_1^\star,\dots, \mU_N^\star)\mId\|.
\end{eqnarray}
Therefore, we can derive
\begin{eqnarray}
\label{expansion of ek final}
&\!\!\!\!\!\!\!\!&\E_{\theta_1,\dots,\theta_N}[g(\mU_1^0,\dots, \mU_N^0)] - g(\mU_1^\star,\dots, \mU_N^\star)\nonumber\\
&\!\!\!\!=\!\!\!\!& e_N \nonumber\\
&\!\!\!\!\leq\!\!\!\!&   \alpha e_{N-1} + \beta\|\mH - g(\mU_1^\star,\dots, \mU_N^\star)\mId\|\nonumber\\
&\!\!\!\!\leq\!\!\!\!&  \alpha^N e_0 + \beta\|\mH - g(\mU_1^\star,\dots, \mU_N^\star)\mId\|\sum_{j=0}^{N-1}\alpha^j\nonumber\\
&\!\!\!\!=\!\!\!\!&\alpha^N e_0 + \beta\|\mH - g(\mU_1^\star,\dots, \mU_N^\star)\mId\| \frac{1-\alpha^N}{1-\alpha}\nonumber\\
&\!\!\!\!=\!\!\!\!&\alpha^N e_0 +  (1 - \alpha^N) \|\mH - g(\mU_1^\star,\dots, \mU_N^\star)\mId\|\nonumber\\
&\!\!\!\!=\!\!\!\!&\alpha^N (\langle\vphi_0| \mH |\vphi_0 \rangle - g(\mU_1^\star,\dots, \mU_N^\star) )+  (1 - \alpha^N) \|\mH - g(\mU_1^\star,\dots, \mU_N^\star)\mId\|\nonumber\\
&\!\!\!\!=\!\!\!\!&\alpha^N (\langle\vphi_0| \mH |\vphi_0 \rangle - g(\mU_1^\star,\dots, \mU_N^\star) )+  (1 - \alpha^N) (\|\mH\| - g(\mU_1^\star,\dots, \mU_N^\star)),
\end{eqnarray}
where the third equality uses $ \beta = 1-\alpha$.

Finally, we apply a concentration inequality to bound the gap $g(\mU_1^0,\dots, \mU_N^0) - g(\mU_1^\star,\dots, \mU_N^\star)$, where each $\mU_i^0 := \mU_i^0(\theta_i)$ is initialized via random parameters $(\theta_1,\dots, \theta_N )$. First, for any $(\theta_1,\dots, \theta_N )$ and $(\theta_1',\dots, \theta_N' )$, we have
\begin{eqnarray}
\label{Lipschi constant of g}
&\!\!\!\!\!\!\!\!&|g(\mU_1^0(\theta_1),\dots, \mU_N^0(\theta_N)) - g(\mU_1^\star,\dots, \mU_N^\star) - (g(\mU_1^0(\theta_1'),\dots, \mU_N^0(\theta_N')) - g(\mU_1^\star,\dots, \mU_N^\star))  |\nonumber\\
&\!\!\!\!=\!\!\!\!&\|\langle \vphi_0 | (\mU_N(\theta_N))^\dagger \cdots (\mU_1(\theta_1))^\dagger \mH \mU_1(\theta_1) \cdots \mU_N(\theta_N) | \vphi_0 \rangle\nonumber\\
&\!\!\!\!\!\!\!\!& - \langle \vphi_0 | (\mU_N(\theta_N'))^\dagger \cdots (\mU_1(\theta_1'))^\dagger \mH \mU_1(\theta_1') \cdots \mU_N(\theta_N') | \vphi_0 \rangle   \|_F\nonumber\\
&\!\!\!\!\leq\!\!\!\!&\| (\mU_N(\theta_N))^\dagger \cdots (\mU_1(\theta_1))^\dagger \mH \mU_1(\theta_1) \cdots \mU_N(\theta_N)  -  (\mU_N(\theta_N'))^\dagger \cdots (\mU_1(\theta_1'))^\dagger \mH \mU_1(\theta_1') \cdots \mU_N(\theta_N')    \|\nonumber\\
&\!\!\!\!\leq\!\!\!\!& 2 \|\mH\| \|\mU_1(\theta_1) \cdots \mU_N(\theta_N) - \mU_1(\theta_1') \cdots \mU_N(\theta_N')\|\nonumber\\
&\!\!\!\!\leq\!\!\!\!& 2 \|\mH\| \sum_{k=1}^{N}\|\mU_k(\theta_k) - \mU_k(\theta_k') \|\nonumber\\
&\!\!\!\!\leq\!\!\!\!& 2 \|\mH\| \sum_{k=1}^{N}|\theta_k - \theta_k' |\nonumber\\
&\!\!\!\!\leq\!\!\!\!& 2 \sqrt{N} \|\mH\| \sqrt{ \sum_{k=1}^{N}|\theta_k - \theta_k' |^2},
\end{eqnarray}
where the fourth inequality uses
\begin{eqnarray}
\label{bound of two unitary matrices}
&\!\!\!\!\!\!\!\!&\|\mU_k(\theta_k) - \mU_k(\theta_k') \|\nonumber\\
&\!\!\!\!=\!\!\!\!& \|\cos(\theta_k)\mId - i\sin(\theta_k)\mP_k - (\cos(\theta_k')\mId - i\sin(\theta_k')\mP_k)  \| \nonumber\\
&\!\!\!\!=\!\!\!\!& \max_{\lambda \in \{\pm 1\}} \big|(\cos(\theta_k) - \cos(\theta_k')) - i(\sin(\theta_k) - \sin(\theta_k'))\lambda \big| \nonumber\\
&\!\!\!\!=\!\!\!\!& \sqrt{(\cos(\theta_k) - \cos(\theta_k')   )^2 + (\sin(\theta_k) - \sin(\theta_k')   )^2 }\nonumber\\
&\!\!\!\!=\!\!\!\!& \sqrt{2 - 2\cos(\theta_k - \theta_k')  }\nonumber\\
&\!\!\!\!=\!\!\!\!&2\bigg|\sin\bigg( \frac{\theta_k - \theta_k'}{2}\bigg) \bigg|\nonumber\\
&\!\!\!\!\leq\!\!\!\!& |\theta_k - \theta_k' |,
\end{eqnarray}
in which we use that $\mP_k$ is Hermitian with eigenvalues $\pm 1$, and hence the spectral norm equals the maximum modulus of its eigenvalues.

Furthermore, based on \eqref{Lipschi constant of g} and the Gaussian concentration inequality in \cite[Eq.~(1)]{fresen2023variations}, we obtain that for any $t > 0$,
\begin{eqnarray}
\label{Gaussian concentration inequality1}
\P{ g(\mU_1^0(\theta_1),\dots, \mU_N^0(\theta_N)) - g(\mU_1^\star,\dots, \mU_N^\star) -  e_N > t }\leq 2 e^{-\frac{t^2}{8N\sigma^2\|\mH\|^2}},
\end{eqnarray}
Let $t = 2\|\mH\|\sigma\sqrt{2N\log(2/\delta) }$. Then, by \eqref{expansion of ek final}, with probability at least $1-\delta$, we have
\begin{eqnarray}
\label{upper bound of gap difference}
g(\mU_1^0(\theta_1),\dots, \mU_N^0(\theta_N)) - g(\mU_1^\star,\dots, \mU_N^\star) &\!\!\!\!\leq\!\!\!\!& \bigg(\frac{1+e^{-2\sigma^2}}{2}\bigg)^N (\langle\vphi_0| \mH |\vphi_0 \rangle - g(\mU_1^\star,\dots, \mU_N^\star) )\nonumber\\
&\!\!\!\!\!\!\!\!& + \bigg(1 - \bigg(\frac{1+e^{-2\sigma^2}}{2}\bigg)^N\bigg) (\|\mH\| - g(\mU_1^\star,\dots, \mU_N^\star))\nonumber\\
&\!\!\!\!\!\!\!\!& + 2\|\mH\|\sigma\sqrt{2N\log(2/\delta) }.
\end{eqnarray}

\section{Proof of \Cref{Theorem: robustness of noisy multi circuit}}
\label{Proof of error bound appendix}

We begin by decomposing the error as
\begin{eqnarray}
\label{noisy objective expansion}
&\!\!\!\!\!\!\!\!&\wh{g}(\mU_1^{t+1},\ldots,\mU_N^{t+1}) - g(\mU_1^\star,\ldots,\mU_N^\star)\nonumber\\
&\!\!\!\! = \!\!\!\!&\wh{g}(\mU_1^{t+1},\ldots,\mU_N^{t+1}) - g(\mU_1^{t+1},\ldots,\mU_N^{t+1}) + g(\mU_1^{t+1},\ldots,\mU_N^{t+1}) - g(\mU_1^\star,\ldots,\mU_N^\star).
\end{eqnarray}
The first term corresponds to the estimation error induced by measurement noise, while the second term captures the optimization error.
To control the optimization error, we invoke \Cref{Theorem: convergence analysis of multi circuit}.
In particular, when the initialization satisfies $g(\mU_1^0,\dots,\mU_N^0) - g(\mU_1^\star,\dots,\mU_N^\star)\leq \frac{\Delta_1}{2}$, and the step size is chosen as $\mu \leq \frac{1}{(8N+1)\|\mH\|}$, the Riemannian gradient descent updates in \eqref{iteration of RGD on the Stiefel N layer} guarantee that
\begin{eqnarray}
\label{descent direction of multi layer quantum circuit final conclusion1 noisy appendix}
 g(\mU_1^{t+1},\ldots,\mU_N^{t+1})- g(\mU_1^\star,\dots,\mU_N^\star)\leq\bigg(1 -\frac{\Delta_1^2\mu}{32\|\mH\|}\bigg)^{t+1}(g(\mU_1^0,\dots,\mU_N^0) - g(\mU_1^\star,\dots,\mU_N^\star) ).
\end{eqnarray}

Next, we bound the estimation error $\wh{g}(\mU_1^{t+1},\ldots,\mU_N^{t+1}) - g(\mU_1^{t+1},\ldots,\mU_N^{t+1})$. To this end, for each Pauli term $\mP_k$, we define i.i.d.\ random variables $X_{k,m} \in \{-1,1\}$ for $m=1,\dots,M$, such that $\E[ X_{k,m}] = p_k^+ - p_k^-$, where $p_k^+$ and $p_k^-$ denote the probabilities of obtaining measurement outcomes $+1$ and $-1$, respectively. Then, the empirical estimate satisfies $\hat{p}_k^+ - \hat{p}_k^- = \frac{\sum_{m=1}^MX_{k,m}}{M}$. Moreover, each $X_{k,m}$ is bounded as $-1 \le X_{k,m} \le 1$. Based on the Hoeffding's inequality in \cite{hoeffding1963probability}, we have
\begin{eqnarray}
\label{concentration inequality of the first term in the noisy setting}
&\!\!\!\!\!\!\!\!& \P{ \wh{g}(\mU_1^{t+1},\ldots,\mU_N^{t+1}) - g(\mU_1^{t+1},\ldots,\mU_N^{t+1}) >t }\nonumber\\
&\!\!\!\! = \!\!\!\!&  \P{  \sum_{k=1}^{L}\frac{\alpha_k\sum_{m=1}^{M}X_{k,m}}{M} -  \sum_{k=1}^{L}\alpha_k( p_k^+ - p_k^-) >t  }\leq e^{-\frac{Mt^2}{2\sum_{k=1}^L \alpha_k^2}}.
\end{eqnarray}
Taking $t =  \sqrt{\frac{2\sum_{k=1}^L \alpha_k^2 \log(1/\gamma)}{M}}$, we obtain that, with probability at least $1-\gamma$,
\begin{eqnarray}
\label{concentration inequality of the first term in the noisy setting final bound}
\wh{g}(\mU_1^{t+1},\ldots,\mU_N^{t+1}) - g(\mU_1^{t+1},\ldots,\mU_N^{t+1}) \leq \sqrt{\frac{2\sum_{k=1}^L \alpha_k^2 \log(1/\gamma)}{M}}.
\end{eqnarray}

Finally, combining \eqref{descent direction of multi layer quantum circuit final conclusion1 noisy appendix} and \eqref{concentration inequality of the first term in the noisy setting final bound}, we obtain that, with probability at least $1-\gamma$,
\begin{eqnarray}
\label{noisy objective expansion final version appendix}
&\!\!\!\!\!\!\!\!&\wh{g}(\mU_1^{t+1},\ldots,\mU_N^{t+1}) - g(\mU_1^\star,\ldots,\mU_N^\star)\nonumber\\
&\!\!\!\! \leq \!\!\!\!& \sqrt{\frac{2\sum_{k=1}^L \alpha_k^2 \log(1/\gamma)}{M}} + \bigg(1 -\frac{\Delta_1^2\mu}{32\|\mH\|}\bigg)^{t+1}(g(\mU_1^0,\dots,\mU_N^0) - g(\mU_1^\star,\dots,\mU_N^\star) ).
\end{eqnarray}

\section{Proof of \Cref{lemma: robustness of noisy multi circuit}}
\label{Proof of optimal allocation}

Following the same analysis of \eqref{concentration inequality of the first term in the noisy setting}, we can derive
\begin{eqnarray}
\label{concentration inequality of the first term in the noisy setting different M}
&\!\!\!\!\!\!\!\!& \P{ \wh{g}(\mU_1^{t+1},\ldots,\mU_N^{t+1}) - g(\mU_1^{t+1},\ldots,\mU_N^{t+1}) >t }\nonumber\\
&\!\!\!\! = \!\!\!\!&  \P{  \sum_{k=1}^{L}\frac{\alpha_k\sum_{m=1}^{M_k}X_{k,m}}{M_k} -  \sum_{k=1}^{L}\alpha_k( p_k^+ - p_k^-) >t  }\leq e^{-\frac{t^2}{2\sum_{k=1}^L \alpha_k^2/M_k}}.
\end{eqnarray}
Taking $t =  \sqrt{2\sum_{k=1}^L \frac{\alpha_k^2}{M_k} \log(1/\gamma)}$, we obtain that, with probability at least $1-\gamma$,
\begin{eqnarray}
\label{concentration inequality of the first term in the noisy setting final bound different M}
\wh{g}(\mU_1^{t+1},\ldots,\mU_N^{t+1}) - g(\mU_1^{t+1},\ldots,\mU_N^{t+1}) \leq \sqrt{2\sum_{k=1}^L \frac{\alpha_k^2}{M_k} \log(1/\gamma)}.
\end{eqnarray}
Finally, combining \eqref{descent direction of multi layer quantum circuit final conclusion1 noisy appendix} and \eqref{concentration inequality of the first term in the noisy setting final bound different M}, we obtain that, with probability at least $1-\gamma$,
\begin{eqnarray}
\label{noisy objective expansion final version appendix different M}
&\!\!\!\!\!\!\!\!&\wh{g}(\mU_1^{t+1},\ldots,\mU_N^{t+1}) - g(\mU_1^\star,\ldots,\mU_N^\star)\nonumber\\
&\!\!\!\! \leq \!\!\!\!& \sqrt{2\sum_{k=1}^L \frac{\alpha_k^2}{M_k} \log(1/\gamma)} + \bigg(1 -\frac{\Delta_1^2\mu}{32\|\mH\|}\bigg)^{t+1}(g(\mU_1^0,\dots,\mU_N^0) - g(\mU_1^\star,\dots,\mU_N^\star) ).
\end{eqnarray}

\section{Proof of Eq.~\eqref{closed-form of constrained optimization main paper}}
\label{Proof of optimal allocation closed form}

By introducing a Lagrange multiplier $\lambda$, the constrained optimization problem in \eqref{minimization of total error under constraints} can be reformulated via the following Lagrangian:
\begin{eqnarray}
\label{the constrained optimization Lagrangian}
\calL(M_1,\dots,M_L,\lambda) = \sum_{k=1}^L \frac{\alpha_k^2}{M_k} + \lambda \Big(\sum_{k=1}^L M_k - M_{\text{tot}}\Big).
\end{eqnarray}
Taking the derivative of the Lagrangian with respect to $M_k$ yields $\frac{\partial \mathcal{L}}{\partial M_k} = - \frac{\alpha_k^2}{M_k^2} + \lambda$. Setting this derivative to zero, we obtain $\frac{\alpha_k^2}{M_k^2} = \lambda$. Thus, the optimal number of measurements satisfies $\wh{M}_k = (\frac{1}{\lambda})^{\frac{1}{2}}|\alpha_k|$. Using the budget constraint $\sum_{k=1}^L \wh{M}_k = \sum_{k=1}^L (\frac{1}{\lambda})^{\frac{1}{2}}|\alpha_k| = M_{\text{tot}}$, it follows that $(\frac{1}{\lambda})^{\frac{1}{2}} = \frac{M_{\text{tot}}}{\sum_{j=1}^L |\alpha_j|}$, and hence the optimal allocation admits the closed-form expression
\begin{eqnarray}
\label{closed-form of constrained optimization}
\wh{M}_k = \frac{|\alpha_k|}{\sum_{j=1}^L |\alpha_j|} M_{\text{tot}}, \ \ k=1,\dots,L.
\end{eqnarray}

\section{Auxiliary Material}
\label{Auxiliary Material}

% ---------------------------------------------------------------
% Auxiliary lemma: retraction bound via projection argument
% ---------------------------------------------------------------
\begin{lemma}
\label{lemma:retraction_bound}
Let $\mU \in \calU(D)$ and $\vxi \in \text{T}_{\mU}\calU(D):=\{\mA\in\C^{D\times D}: \mA^\dagger\mU+\mU^\dagger \mA={\bm 0} \}$. The polar decomposition retraction $\text{Retr}_{\mU}(\vxi) = (\mU+\vxi)\bigl((\mU+\vxi)^\dagger(\mU+\vxi)\bigr)^{-1/2}$ satisfies
\begin{equation}
    \|\text{Retr}_{\mU}(\vxi) - (\mU+\vxi)\|_F \leq \frac{1}{2}\|\vxi\|_F^2.
\end{equation}
\end{lemma}

\begin{proof}
Since $\vxi \in \text{T}_{\mU}\calU(D)$, we have $(\mU+\vxi)^\dagger(\mU+\vxi) = \mId + \vxi^\dagger\vxi$. Let $\mU + \vxi = \mV\mSigma\mW^\dagger$ be the SVD of $\mU+\vxi$, so that $\text{Retr}_{\mU}(\vxi) = \mV\mW^\dagger$.
Then
\begin{equation}
    \|\text{Retr}_{\mU}(\vxi) - (\mU+\vxi)\|_F
    = \|\mV\mW^\dagger - \mV\mSigma\mW^\dagger\|_F
    = \|\mId - \mSigma\|_F.
\end{equation}
Since $(\mU+\vxi)^\dagger(\mU+\vxi) = \mId + \vxi^\dagger\vxi$, the singular values $\sigma_k$ of $\mU+\vxi$ satisfy $\sigma_k^2 = 1 + \lambda_k$ where $\lambda_k \geq 0$ are eigenvalues of $\vxi^\dagger\vxi$.
Therefore $\sigma_k \geq 1$ and
\begin{equation}
    \sigma_k - 1
    = \frac{\sigma_k^2 - 1}{\sigma_k + 1}
    \leq \frac{\sigma_k^2 - 1}{2}
    = \frac{\lambda_k}{2}.
\end{equation}
Summing over all $k$:
\begin{equation}
    \|\mId - \mSigma\|_F^2
    = \sum_k(\sigma_k - 1)^2
    \leq \frac{1}{4}\sum_k\lambda_k^2
    = \frac{1}{4}\|\vxi^\dagger\vxi\|_F^2
    \leq \frac{1}{4}\|\vxi\|_F^4.
\end{equation}
Taking square roots gives $\|\mId - \mSigma\|_F \leq \frac{1}{2}\|\vxi\|_F^2$, which completes the proof.
\end{proof}

% ---------------------------------------------------------------
% Main descent lemma
% ---------------------------------------------------------------
\begin{lemma}
\label{lemma:Riemannian descent conclusion}
Let $f(\mU_1,\ldots,\mU_N) = \trace(\mH\cdot \mU_1\cdots \mU_N\vrho_0 \mU_N^\dagger\cdots \mU_1^\dagger)$ with $\mU_h \in \calU(D)$. Suppose $\text{grad}_{\mU_h}f$ is $C$-Lipschitz continuous, i.e., for all $(\mU_1,\ldots,\mU_N)$ and $(\mV_1,\ldots,\mV_N)$,
\begin{equation}
    \|\text{grad}_{\mU_h}f(\mV_1,\ldots,\mV_N)
      - \text{grad}_{\mU_h}f(\mU_1,\ldots,\mU_N)\|_F
    \leq C\sqrt{\sum_{h=1}^{N}\|\mV_h - \mU_h\|_F^2}.
\end{equation}
Let $\vxi_h \in \text{T}_{\mU}\calU(D)$ and $\mW_h = \text{Retr}_{\mU_h}(\vxi_h)$ for each $h$. Then, we have
\begin{equation}
\label{eq:descent}
    f(\mW_1,\ldots,\mW_N)
    \leq f(\mU_1,\ldots,\mU_N)
    + \sum_{h=1}^N \langle \text{grad}_{\mU_h}f(\mU_1,\ldots,\mU_N),\, \vxi_h \rangle
    + \bigg(C\sqrt{N} + \frac{\|\mH\|}{2}\bigg)
      \sum_{h=1}^{N}\|\vxi_h\|_F^2.
\end{equation}
\end{lemma}

\begin{proof}
Let $\mW_h = \text{Retr}_{\mU_h}(\vxi_h)$ for each $h$.
We decompose:
\begin{align}
    &f(\mW_1,\ldots,\mW_N) - f(\mU_1,\ldots,\mU_N)
     - \sum_{h=1}^N\langle \text{grad}_{\mU_h}f(\mU_1,\ldots,\mU_N),\, \vxi_h \rangle
    \nonumber\\
    &= \underbrace{f(\mW_1,\ldots,\mW_N)
       - f(\mU_1+\vxi_1,\ldots,\mU_N+\vxi_N)}_{\mathrm{Term\ I}}
    \nonumber\\
    &\quad
     + \underbrace{f(\mU_1+\vxi_1,\ldots,\mU_N+\vxi_N)
       - f(\mU_1,\ldots,\mU_N)
       - \sum_{h=1}^N\langle \text{grad}_{\mU_h}f(\mU_1,\ldots,\mU_N),\,\vxi_h
       \rangle}_{\mathrm{Term\ II}}.
\end{align}

\textbf{Bounding Term I.}
By the mean value inequality and Lemma~\ref{lemma:retraction_bound}, we have
\begin{align}
\label{eq:descent first term}
    \mathrm{Term\ I}
    &= \sum_{h=1}^{N}
\big[
f(\mW_1,\ldots,\mW_h,\mU_{h+1}+\vxi_{h+1},\ldots,\mU_N+\vxi_N)
\nonumber\\
&- f(\mW_1,\ldots,\mW_{h-1},\mU_h+\vxi_h,\mU_{h+1}+\vxi_{h+1},\ldots,\mU_N+\vxi_N)
\big] \nonumber\\
    &\leq \sum_{h=1}^N
      f_{\max}
      \cdot \|\mW_h - (\mU_h+\vxi_h)\|_F
    \nonumber\\
    &\leq  \frac{\|\mH\|}{2}\sum_{h=1}^N\|\vxi_h\|_F^2,
\end{align}
where the first inequality uses $\|\nabla_{\mU_h}f(\mU_1,\ldots,\mU_N)\|_F = \|(\mU_1\cdots \mU_{h-1})^\dagger\mH\cdot \mU_1\cdots \mU_N\vrho_0(\mU_{h+1}\cdots \mU_{N})^\dagger\|_F \leq \|\vrho_0\|_F\|\mH\| = \|\mH\|=: f_{\max}$ since each $\mU_h$ is unitary and $\|\vrho_0\|_F=1$. Moreover, $\|\mW_h-(\mU_h+\vxi_h)\|_F \leq \frac{1}{2}\|\vxi_h\|_F^2$ follows from Lemma~\ref{lemma:retraction_bound}.

\textbf{Bounding Term II.}
By the fundamental theorem of calculus and telescoping, we can derive
\begin{align}
\label{eq:descent second term}
    \mathrm{Term\ II}
    &= \sum_{h=1}^N \int_0^1
    \Bigl\langle
        \text{grad}_{\mU_h}f(\mU_1+\vxi_1,\ldots,\mU_h+t\vxi_h,\ldots,\mU_N)
        - \text{grad}_{\mU_h}f(\mU_1,\ldots,\mU_N),\,\vxi_h
    \Bigr\rangle\, \mathrm{d}t
    \nonumber\\
    &\leq \sum_{h=1}^N \int_0^1
    \bigl\|\text{grad}_{\mU_h}f(\mU_1+\vxi_1,\ldots,\mU_h+t\vxi_h,\ldots,\mU_N)
    - \text{grad}_{\mU_h}f(\mU_1,\ldots,\mU_N)\bigr\|_F
    \|\vxi_h\|_F\, \mathrm{d}t
    \nonumber\\
    &\leq \sum_{h=1}^N \|\vxi_h\|_F \int_0^1
    C\|(\vxi_1,\ldots,\vxi_{h-1},t\vxi_h,\mathbf{0},\ldots,\mathbf{0})\|\,
    \mathrm{d}t
    \nonumber\\
    &\leq \sum_{h=1}^N \|\vxi_h\|_F
      \cdot C\sqrt{\sum_{j=1}^N\|\vxi_j\|_F^2}
    \nonumber\\
    &\leq C\sqrt{N}\sum_{h=1}^{N}\|\vxi_h\|_F^2,
\end{align}
where the third inequality uses
$\|(\vxi_1,\ldots,\vxi_{h-1},t\vxi_h,\mathbf{0},\ldots,\mathbf{0})\| \leq \sqrt{\sum_{j=1}^N\|\vxi_j\|_F^2}$, and the last line follows from the Cauchy--Schwarz inequality:
$\sum_{h=1}^N\|\vxi_h\|_F \leq \sqrt{N}\sqrt{\sum_{h=1}^N\|\vxi_h\|_F^2}$.

\textbf{Conclusion.} Combining Term I in \eqref{eq:descent first term} and Term II in \eqref{eq:descent second term}, we have
\begin{equation}
    f(\mW_1,\ldots,\mW_N)
    \leq f(\mU_1,\ldots,\mU_N)
    + \sum_{h=1}^N \langle \text{grad}_{\mU_h}f(\mU_1,\ldots,\mU_N),\, \vxi_h \rangle
    + \bigg(C\sqrt{N} + \frac{\|\mH\|}{2}\bigg)
      \sum_{h=1}^{N}\|\vxi_h\|_F^2.
\end{equation}
\end{proof}

%{
%%\bibliographystyle{alpha}
%\bibliographystyle{unsrt}
%\bibliography{reference1}
%}

\end{document}